\documentclass[preprint,showpacs,preprintnumbers,amsmath,amssymb]{revtex4-1}

\usepackage{graphicx}% Include figure files
\usepackage{dcolumn}% Align table columns on decimal point
\usepackage{bm}% bold math

\usepackage{subfiles} %so we can break into multiple .tex files
\usepackage{tikz}
\usetikzlibrary{arrows,automata}
\usepackage{dsfont}
\usepackage{caption}  % use for side-by-side figures
\usepackage{subcaption} % use for side-by-side figures
\usepackage{natbib}
\usepackage{amsthm}

\newtheorem{theorem}{Theorem}[section]
\newtheorem*{theorem*}{Theorem}
\newtheorem{lemma}{Lemma}[section]

\newtheorem{definition}[theorem]{Definition}

\newtheorem{corollary}[theorem]{Corollary}
\newtheorem*{corollary*}{Corollary}

\begin{document}

\preprint{APS/123-QED}

\title{Role of subgraphs in epidemics over finite-size networks under the scaled SIS process}% Force line breaks with \\

\author{June Zhang}%
\email{junez@andrew.cmu.edu}
\author{Jos\'{e} M.F.~Moura}
\altaffiliation[]{Was a visiting professor with New York University and the Center for Urban Science and Policy (CUSP) in 2013-2014}%Lines break automatically or can be forced with \\

\affiliation{%
Carnegie Mellon University\\
Electrical and Computer Engineering Dept.\\
Pittsburgh, PA, USA
}%

\date{\today}% It is always \today, today,

\begin{abstract}
%An article usually includes an abstract, a concise summary of the work
%covered at length in the main body of the article. It is used for
%secondary publications and for information retrieval purposes. Valid
%PACS numbers may be entered using the \verb+\pacs{#1}+ command.

In previous work, we developed the scaled SIS process, which models the dynamics of SIS epidemics over networks. With the scaled SIS process, we can consider networks that are finite-sized and of arbitrary topology (i.e., we are not restricted to specific classes of networks). We derived for the scaled SIS process a closed-form expression for the time-asymptotic probability distribution of the states of all the agents in the network. This closed-form solution of the equilibrium distribution explicitly exhibits the underlying network topology through its adjacency matrix. This paper determines which network configuration is the most probable. We prove that, for a range of epidemics parameters, this combinatorial problem leads to a submodular optimization problem, which is \emph{exactly} solvable in polynomial time. We relate the most-probable configuration to the network structure, in particular, to the existence of high density subgraphs. Depending on the epidemics parameters, subset of agents may be more likely to be infected than others; these more-vulnerable agents form subgraphs that are denser than the overall network. We illustrate our results with a 193 node social network and the 4941 node Western US power grid under different epidemics parameters.

\end{abstract}

\pacs{Valid PACS appear here}% PACS, the Physics and Astronomy
\maketitle

\section{Introduction}
A network is a graph; it is a collection of nodes connected by edges. Networks have been used in science and engineering to represent systems of multiple interconnected, interdependent components. As a result, the network structure has a large impact on the behavior of the system. Quantifying how network structure impacts network function, that is, the behavior of dynamical processes on networks, is a difficult problem since the system components do not behave independently. 

In this paper, we focus on analyzing the behavior of network diffusion processes such as epidemics. Analytical results for epidemics on networks have been obtained under particular conditions: full mixing models (i.e., the underlying network is a complete graph); infinite-sized networks models using mean-field approximation; or for scaled-free networks \cite{newman2010networks, csermely2013structure, pastor2002epidemic, de2010stochastic}. These approaches approximate the underlying network topology with mathematically simpler structures, because accounting for the exact graph topology is a combinatorial problem that is difficult to analyze and computationally expensive to compute. We showed in previous work \cite{JZhangJournal, JZhang2} that ,for a specific network diffusion process, which we called the \emph{scaled SIS} (Susceptible-Infected-Susceptible) process, it is possible to characterize its time-asymptotic behavior on any arbitrary, finite-sized network with $N$ agents. 

The scaled SIS process is Markov. It accounts for 1) exogenous (i.e., spontaneous) infection at rate $\lambda$; 2) endogenous (i.e., neighbor-to-neighbor) infection at rate $\gamma$; and 3) healing at rate $\mu$. The time-asymptotic behavior of the process is described by its equilibrium distribution, which is a PMF (probability mass function) over all $2^N$ possible network configurations. Our approach preserves the full microscopic states of all the agents in contrast to previous approaches that only provide results for aggregate or macroscopic states (e.g., fraction of infected agents) \cite{1238052}. However, retaining the exact network configuration means that the computational complexity of solving for the equilibrium distribution, an eigenvector problem, scales exponentially with the size of the network, $N$.

We have shown that, under specific assumptions on the form of the endogenous infection, the scaled SIS process is a \emph{reversible} Markov process for which we can find its equilibrium distribution in closed form, avoiding solving a large eigenvalue/eigenvector problem. Further, the equilibrium distribution that we derived exhibits explicitly the underlying network structure through the network adjacency matrix. The equilibrium distribution is parameterized by two parameters: $\left(\frac{\lambda}{\mu}, \gamma \right)$, where as usual, parameter $\frac{\lambda}{\mu}$ controls the exogenous, or the topology-independent behavior of the scaled SIS process, whereas parameter $\gamma$ controls the endogenous or the topology-dependent behavior of the process. 

We used the equilibrium distribution to address the question of which of the $2^N$ possible configurations in a network is the most likely to occur in the long run. We refer to this as the most-probable configuration, which is found by maximizing the equilibrium distribution. This optimization (called the Most-Probable Configuration Problem) is difficult because: 1) it is combinatorial; 2) it depends on the healing/infection parameters of the scaled SIS process; and 3) it depends on the underlying network topology. Previously in \cite{JZhangJournal}, we partitioned the space of $\left(\frac{\lambda}{\mu}, \gamma \right)$ values into four regimes and were able to find the most-probable configuration in Regime II) \textbf{Endogenous Infection Dominant}, for which $0 < \frac{\lambda}{\mu} \leq 1, \gamma > 1$, for only specific types of networks: $k$-regular, complete multipartite, and complete multipartite with $k$-regular islands. We showed for these specific networks that the most-probable configuration solution space exhibits phase transition behavior depending on the network structure and epidemics parameters. 

This paper considers the Most-Probable Configuration Problem for arbitrary networks. We are able to prove that this leads to the optimization of a submodular function for which we have a polynomial time solution. Further, we show which clusters of agents in the network are more vulnerable to epidemics than others. These are relevant questions in applications. For example, these are the clusters to focus on in marketing campaigns or when combating epidemics. 

We review the scaled SIS process in Section~\ref{sec:model} and set up the Most-Probable Configuration Problem in Section~\ref{sec:xstar}. In Section~\ref{sec:submodular}, we show that, in Regime II), the Most-Probable Configuration Problem can be transformed into an equivalent submodular problem, and that it is possible to solve for its \emph{exact} solution in \emph{polynomial time}. We apply this to solve the most-probable configuration for two example networks: the 193 node acquaintance network of drug users in Hartford, CT \cite{weeks2002social}, and the 4941 node network of the Western US power grid \cite{watts1998collective}. Section~\ref{sec:netstruct} shows how the solution space of the Most-Probable Configuration Problem in Regime II) relates to the density of subgraphs in the network. Section~\ref{sec:conclusion} concludes the paper.

%%%%%%%%%%

\section{Scaled SIS Process}\label{sec:model}
Consider a population of $N$ agents whose interconnections are represented by a static, simple, unweighted, undirected, connected graph, $G(V,E)$, where $V(G)$ is the set of vertices and $E(G)$ is the set of edges. For background on graphs see \cite{algraph}. The topology of $G$ is captured by the symmetric $N \times N$ adjacency matrix, $A$. The state of the $i$\textsuperscript{th} agent is denoted by $x_i$. Agents can be in one of two states: susceptible ($x_i =0$) or infected ($x_i = 1$); susceptible agents are vulnerable to infections since there is no immunization in the system.

Let
\[
\mathbf{x} = [x_1, x_2, \ldots, x_{N}]^T.
\]
We will refer to $x_i$ as the \emph{agent state} and $\mathbf{x}$ as either the \emph{network state} or the \emph{network configuration}. The configuration state space is $\mathcal{X} = \{\mathbf{x}\}$, with cardinality $\left| \mathcal{X} \right | = 2^N$.

The scaled SIS process models the evolution of the network state, $\mathbf{x}$, over time according to the stochastic microscopic interaction rules from the SIS (susceptible-infected-susceptible) epidemics. The SIS framework assumes that infected agents can heal and become reinfected so it does not account for immunization \cite{newman2010networks}. Let $X(t) = \mathbf{x}$ be the state of the network at time $t, \, t \geq 0$. Under appropriate assumptions, $X(t)$ is a continuous-time Markov process \cite{Draief:2006:TVS:1190095.1190160, Ganesh, 1238052}. The scaled SIS process accounts for 1) exogenous infection (i.e., susceptibles spontaneously develop infection); 2) endogenous infection (i.e., susceptibles become infected due to infection from infective neighbors); and 3) healing events. These processes are independent. At time $t$, only one one agent is affected. By including both exogenous infection and healing, the scaled SIS process does \emph{not} have an absorbing state at equilibrium. 

The scaled SIS process is Markov; each network state is a state of the Markov process. We define two operators on the network state, $\mathbf{x} = [x_1, x_2, \ldots x_i,\ldots x_j, \ldots, x_{N}]^T$. We use the following notation:
\begin{eqnarray*}
H_i \mathbf{x} =  [x_1, x_2, \ldots, x_i = 1, \ldots,  x_{N}]^T \\
H_{j \bullet }\mathbf{x} =  [x_1, x_2, \ldots, x_j = 0, \ldots,  x_{N}]^T.
\end{eqnarray*}
The operator $H_i$ defines the operation that agent $i$ becomes infected. If agent $i$ is already infected, the operator does nothing. The operator $H_{j \bullet }$ defines the operation that agent $j$ is healed. If agent $j$ is already uninfected, the operator does nothing.

The time the process spends in a particular state is random and exponentially distributed, with the following transition rates corresponding to infection and healing events, respectively:

\begin{enumerate}
\item $X(t)$ jumps to the network state where the $i$th agent, which was healthy, becomes infected with transition rate
\begin{equation}\label{eq:qnTk}
q(\mathbf{x}, H_i\mathbf{x}) = \lambda \gamma^{d_i}, \quad \mathbf{x} \neq H_i\mathbf{x},
\end{equation}
where $d_i = {\sum_{j=1}^{N} \mathds{1}(x_j = 1) A_{ij}}$, is the number of infected neighbors of node $i$. The symbol $\mathds{1}(\cdot)$ is the indicator function, and $A=\left[A_{ij}\right]$ is the adjacency matrix of the arbitrary network $G$ that captures the interactions among the agents. There are two components to the infection rate. If the $i$th agent has no infected neighbors, $d_i = 0$, and the transition rate reduces to $\lambda > 0$. We interpret $\lambda$ as the exogenous infection rate, the rate a susceptible agent spontaneously becomes infected; it is the same for all the agents in the network. If the $i$th agent has $d_i$ infected neighbors, the infective rate is $\lambda\gamma^{d_i}$; it is the product of $\lambda$ and the endogenous infection rate, $\gamma >0$, \emph{scaled} by $d_i$, the number of infected neighbors of agent $i$. Because  of this factor, the infective rate depends on the network topology. 

\item $X(t)$ jumps to the network state where the $j$th agent, which was infected, heals with transition rate:
\begin{equation}\label{eq:qnTj}
q(\mathbf{x}, H_{j \bullet }\mathbf{x}) = \mu, \quad \mathbf{x} \neq H_{j \bullet }\mathbf{x}.
\end{equation}
The healing rate, $\mu > 0$, is the same for all the agents in the system. 
\end{enumerate}

\subsection{Equilibrium Distribution}
The evolution of the scaled SIS process is captured by the rate (infinitesimal) matrix $\mathbf{Q}$ of the Markov process $X(t)$. The assumption that the underlying network $G$ is connected assures that the Markov process is irreducible. Therefore, the equilibrium distribution, $\pi(\mathbf{x})$, exists and is given by the left eigenvector corresponding to the 0 eigenvalue of $\mathbf{Q}$, the rate matrix \cite{norris1998markov}. The problem in determining the equilibrium distribution $\pi(\mathbf{x})$ is that its computation is prohibitively expensive for meaningful sized networks since $\mathbf{Q}$ is a $2^N \times 2^N$ matrix. This has limited the analysis of epidemics and spreading processes on networks to either: 1) full mixing models (e.g., where every agent comes in contact with every other agents \textemdash the network is a complete graph); 2) to small scale simulations, where $N$ is small so that $O((2^N)^3)$ operations are feasible; or 3) to mean field type approximations of special network configurations.

We proved in \cite{JZhang2}, see also \cite{JZhangJournal}, that the scaled SIS process is a \emph{reversible} Markov process by showing that its equilibrium distribution satisfies not only the global balance equation but also the detailed balance equation \cite{Kelly}. For reversible Markov processes, the equilibrium distribution is unique. We derived the equilibrium distribution of the scaled SIS process to be:

\begin{equation} \label{eq:equilibriumdistribution}
\pi({\bf x}) =\frac{1}{Z}\left( \frac{\lambda}{\mu}\right)^{1^T{\bf x}}  \gamma^{\frac{{\bf x}^TA{\bf x}}{2}  }, \quad  \mathbf{x}, \in \mathcal{X}
\end{equation}
where $Z$ is the partition function,
\begin{equation} \label{eq:partitionfunction}
Z=  \sum_{{\bf x} \in \mathcal{X}}  \left( \frac{\lambda}{\mu}\right)^{1^T{\bf x}}  \gamma^{\frac{{\bf x}^TA{\bf x}}{2}}.   
\end{equation}

Previous epidemics model call the ratio $\frac{\lambda}{\mu}$, the \emph{effective infection rate} \cite{PhysRevE.90.012810}. The equilibrium distribution, $\pi(\mathbf{x})$, factors as the product of three terms: 1) the normalization by the partition function; 2) the term $\left( \frac{\lambda}{\mu}\right)^{1^T{\bf x}}$ that is topology independent since the exogenous infection rate $\lambda$ and the healing rate $\mu$ are identical for all the agents in the network, and the total number of infected agents, $1^T{\bf x}$, does not depend on the topology; and 3) the $\gamma^{\frac{{\bf x}^TA{\bf x}}{2}}$ that explicitly accounts for the exact network through its adjacency matrix~$A$. It is topology dependent since the endogenous infection rate $\gamma$ is scaled by the number of infected neighbors; the number of edges where both end nodes are infected (we call them \emph{infected edges}), $\frac{{\bf x}^TA{\bf x}}{2}$, explicitly depends on the adjacency matrix of the underlying network.

\subsection{Parameter Regimes}
The scaled SIS Process can model different types of network diffusion processes depending on the values of the rate parameters; in particular, if the effective exogenous infection rate, $\frac{\lambda}{\mu}$, and the endogenous infection rate, $\gamma$, are between 0 and 1, or if they are greater than 1. In \cite{JZhangJournal}, we identified 4 regimes. 

When both parameters are either between 0 and 1 or greater than 1, then the most-probable configuration is either the $\mathbf{x}^0 = [0,0 \ldots, 0]^T$ configuration or the $\mathbf{x}^N = [1,1 \ldots, 1]^T$ configuration. Reference \cite{JZhangJournal} also investigated Regime III) where $\frac{\lambda}{\mu} >1, 0 < \gamma \leq 1$. This regime models the counter-intuitive behavior where an increasing number of infected agents delays additional infection in the network. In this paper, we focus our analysis on Regime II) \textbf{Endogenous Infection Dominant:} $0 < \frac{\lambda}{\mu} \leq 1, \gamma> 1$. Regime II best models epidemics and similar types of spreading processes. 

The effective exogenous infection rate, $\frac{\lambda}{\mu}$, indicates the preference of individual agents. With $0 < \frac{\lambda}{\mu} \leq 1$, the healing rate is larger than the exogenous infection rate; agents prefer the healthy state to the infected state. With $\gamma >1$, however, additional infected neighbors increase the rate at which the healthy agent becomes infected; thereby the network helps to spread the infection. As a result, the network topology is crucial to determine the behavior of the scaled SIS process at equilibrium.

In the next section, we introduce the Most-Probable Configuration Problem, which solves for the network configurations with maximum equilibrium probability. Because there is \emph{competition} between the topology independent term and the topology dependent term, the most-probable configuration exhibits complex phase transition behavior depending on the effective exogenous infection rate $\frac{\lambda}{\mu}$, the endogenous infection rate $\gamma$, and the underlying network topology.

%%%%%%%%%%%%%%%%

\section{Most-Probable Configuration Problem}\label{sec:xstar}

In the previous section, we showed that, for the scaled SIS process, we are able to derive its equilibrium distribution, $\pi(\mathbf{x})$, analytically, see equation~\eqref{eq:equilibriumdistribution}. The equilibrium distribution describes the long-run behavior of the network epidemics. While the partition function~\eqref{eq:partitionfunction} renders the exact calculation of the equilibrium distribution infeasible for meaningful size networks, knowing the equilibrium distribution expression allows us to quickly compare between network configurations, addressing, for example questions like which of the two is more probable. Of all the possible $2^N$ network configurations, one is of particular interest, namely, the configuration of infected and healthy agents that has a higher chance of occurring in the long run. This is the configuration $\mathbf{x}^*$ that maximizes $\pi(\mathbf{x})$. Formally, $\mathbf{x}^*$ maximizes the equilibrium probability: 
\begin{equation}\label{eq:xstar}
\mathbf{x}^* = \arg \max_{\mathbf{x} \in \mathcal{X}} \pi(\mathbf{x}) =  \arg \max_{\mathbf{x} \in \mathcal{X}}\left( \frac{\lambda}{\mu}\right)^{1^T{\bf x}}  \gamma^{\frac{{\bf x}^TA{\bf x}}{2}  }.  
\end{equation}

We call this the Most-Probable Configuration Problem and $\mathbf{x}^*$ the \emph{most-probable configuration}. The Most-Probable Configuration Problem is a combinatorial optimization problem as agents can only be in one of two states; its solution is dependent on the effective exogenous infection rate, $\frac{\lambda}{\mu}$, the endogenous infection rate $\gamma$, and the underlying network topology, captured by the adjacency matrix, $A$.

Previously in \cite{JZhangJournal}, we provided analytical results for the Most-Probable Configuration Problem in Regime II) \textbf{Endogenous Infection Dominant:} $0 < \frac{\lambda}{\mu} \leq 1, \gamma> 1$ for particular networks, namely, structured network topologies such as $k$-regular, complete multipartite, complete multipartite with $k$-regular islands. We observed a phase transition behavior. Below a threshold condition that depends on the parameters $\left(\frac{\lambda}{\mu}, \gamma \right)$ and on the network topology, the most-probable configuration is $\mathbf{x}^0 = [0, 0, \ldots, 0]$, the configuration where all agents are susceptibles. Above the threshold condition, the most-probable configuration is $\mathbf{x}^N = [1, 1, \ldots, 1]$, the configuration where all agents are infected. 

This paper extends the analysis of the Most-Probable Configuration Problem in Regime II) to \emph{arbitrary} network topologies. We will show that, for arbitrary networks, the most-probable configuration may be configurations other than $\mathbf{x}^0$ and $\mathbf{x}^N$. We call these solutions to the Most-Probable Configuration Problem \emph{non-degenerate configurations}. These solutions are useful for identifying agents and communities that are more vulnerable to the epidemics. We will relate these communities to the structure of the networks in detail later. Figure~\ref{fig:KPP} and Figure~\ref{fig:Power} show the most-probable configurations obtained by the method of Section~\ref{sec:submodular} for two example networks: a 193-node acquaintance network \cite{weeks2002social} and the 4941-node power grid \cite{watts1998collective}. These are non-degenerate configurations where only a subset of agents are infected.

In Section~\ref{sec:submodular}, we prove that we can solve \emph{exactly} for the most-probable configuration in Regime II) in polynomial time using submodular optimization. Then, in Section~\ref{sec:netstruct}, we discuss the relationship between the most-probable configuration and the network topology, in particular, the relation between non-degenerate configurations and network topology.

%%%%%%%%%%%%%%%%%%

\section{Submodularity and the Most-Probable Configuration}\label{sec:submodular}

In this section, we solve the Most-Probable Configuration Problem in Regime II in polynomial time by showing that the problem can be transformed into a submodular function. First, we review the definition of submodular functions.

\subsection{Submodular Function} 

The Most-Probable Configuration Problem is the maximization of a pseudo-Boolean function. Pseudo-Boolean functions are functions that map $N$ binary variables to a real number \cite{billionnet1985maximizing}. Minimization of general pseudo-Boolean functions is NP-hard \cite{boros2002pseudo}. Gr{\"o}tschel, Lov{\'a}sz, and Schrijver, \cite{grotschel1981ellipsoid}, proved that the minimization of a pseudo-Boolean function that is submodular can be done in polynomial time. If the function is supermodular, its maximization is in polynomial time. 

A pseudo-Boolean function, $f: \{0,1\}^N \to \mathcal{R}$, is also a set function $g: \mathcal{P}(V) \to \mathcal{R}$ where $\mathcal{P}(V)$ is the power set of $V = \{1,2,\ldots, N \}$. There are many equivalent definitions of submodularity \cite{lovasz1983submodular}. The one we use in this paper is the following:

\begin{definition}[\cite{billionnet1985maximizing}]
A set function, $g: \mathcal{P}(V) \to \mathcal{R}$, is submodular if and only if for any $\alpha_1 \subseteq V, \alpha_2 \subseteq \alpha_1, i \in V \setminus \alpha_1$:
\[
g(\alpha_1 \cup \{i\}) - g(\alpha_1) \leq g(\alpha_2 \cup \{i\}) - g(\alpha_2).
\]
\end{definition}

For a submodular function, the incremental gain of adding an element to the set $\alpha_1$ is less than or equal to the gain of adding the element to a smaller subset of $\alpha_1$. A supermodular function has the inequality in the opposite direction.

\subsection{Most-Probable Configuration: A Submodular Problem}

The Most-Probable Configuration Problem \eqref{eq:xstar} seeks the maximum of a pseudo-Boolean function that maps a 0-1 vector, the network configuration $\mathbf{x}$, to a scalar. The network configuration $\mathbf{x} \in \{0,1\}^N$ is the characteristic vector or characteristic function of the set of infected agents: $\alpha_{\mathbf{x} }= \{ i \mid i \in V, x_i = 1\}$. Let $h(\alpha_{\mathbf{x} })$ be the set of infected edges (i.e., edges where both end nodes are infected) in configuration $\mathbf{x}$: $h(\alpha_{\mathbf{x} }) = \{\{i,j\} \mid i,j \in V, x_i = 1, x_j = 1\}$.

The number of infected agents in configuration $\mathbf{x}$ is $\left\vert{\alpha_{\mathbf{x} }}\right\vert = 1^T\mathbf{x}$. The number of infected edges is $\left\vert h(\alpha_{\mathbf{x} })\right\vert = \frac{\mathbf{x}^TA\mathbf{x}}{2}$. The Most-Probable Configuration Problem is then to solve for the maximum argument of 
\begin{equation}\label{eq:xstarset}
g(\alpha_{\mathbf{x}}) = \left( \frac{\lambda}{\mu}\right)^{\left\vert{\alpha_{\mathbf{x} }}\right\vert}  \gamma^{\left\vert h(\alpha_{\mathbf{x} })\right\vert}.
\end{equation}

We will prove in Theorem~\ref{thm:submodular} that $-\log(g(\alpha_{\mathbf{x}}))$ is a submodular function. Therefore, we can solve for its minimum argument in polynomial time. Lemma~\ref{lemma:submodular} sets up some basic conditions that makes proving Theorem~\ref{thm:submodular} easier.

\begin{lemma}\label{lemma:submodular}
Consider two sets of infected agents, $\alpha_1, \alpha_2 \subseteq V$ and $ i \in V \setminus \alpha_1$. The cardinalities of $\alpha_1$ and $\alpha_2$ are $\left\vert{\alpha_1 }\right\vert = n_1$ and 
$\left\vert{\alpha_2 }\right\vert = n_2$, respectively; then $\left\vert{\alpha_1 \cup \{i\}}\right\vert = n_1 + 1$, and $\left\vert{\alpha_2 \cup \{i\}}\right\vert = n_2 + 1$. The numbers of infected edges induced by $\alpha_1$ and $\alpha_2$ are $\left\vert h(\alpha_1)  \right\vert = e_1$ and $\left\vert h(\alpha_2)  \right\vert = e_2$, respectively. Let $\left\vert h(\alpha_1 \cup \{i\}) \right\vert = e_1 + m_1$ and $\left\vert h(\alpha_2 \cup \{i\}) \right\vert = e_2 + m_2$; therefore $m_1$ is the number of additional infected edges created with the inclusion of agent $i$ in $\alpha_1$ and $m_2$ is the number of additional infected edges created with the inclusion of agent $i$ in $\alpha_2$. Let $\alpha_2 \subseteq \alpha_1$. Then:

\begin{enumerate}
\item $n_1 \ge n_2$.
\item $e_1 \geq e_2$.
\item $m_1 \geq  m_2$.
\end{enumerate}
\end{lemma}

\begin{proof}\label{proof:lemmasubmodular}
\begin{enumerate}
\item When $\alpha_2 \subset \alpha_1$, $\alpha_2$ must have strictly fewer number of infected agents than $\alpha_1$. When $\alpha_2 = \alpha_1$, then they contain the same number of infected agents. Hence, $n_1 \ge n_2$.

\item When $\alpha_2 \subset \alpha_1$, infected agents in $\alpha_2$ can not induce more infected edges than the number of infected edges induced by the infected agents in $\alpha_1$. When $\alpha_2 = \alpha_1$, then the infected agents in $\alpha_1$ and $\alpha_2$ will induce the same number of infected edges. Hence, $e_1 \geq e_2$.

\item Every infected agent in $\alpha_2$ is an infected agent in $\alpha_1$. Every new infected edge connecting the infected agent $j \in \alpha_2$ with $i$ is also a new infected edge in $\alpha_1 \cup \{i\}$. However, some edge may also have $j \in \alpha_1$. Hence, $m_1 \ge m_2$.

\end{enumerate}
\end{proof}

\begin{theorem}\label{thm:submodular}
Let $g(\alpha_{\mathbf{x}})$ be the set function given in \eqref{eq:xstarset}. If $\lambda >0, \mu> 0$ and $\gamma \geq 1$, then $-\log(g(\alpha_{\mathbf{x}}) )$ is a submodular function, where
\[
-\log(g(\alpha_{\mathbf{x}}) ) = -\left\vert{\alpha_{\mathbf{x} }}\right\vert \log\left(\frac{\lambda}{\mu} \right) - \left\vert h(\alpha_{\mathbf{x} })\right\vert \log(\gamma).
\]
\end{theorem}

\begin{proof}\label{proof:theoremsubmodular}
To prove submodularity of $-\log(g(\alpha_{\mathbf{x}}) )$, we need to show that
\begin{align}\label{eq:toprove}
-\log(g(\alpha_1 \cup \{i\})) + \log(g(\alpha_1)) \leq -\log(g(\alpha_2 \cup \{i\})) + \log(g(\alpha_2)), 
\end{align}
for any $\alpha_1 \subseteq V, \alpha_2 \subseteq \alpha_1, i \in V \setminus \alpha_1$.

The left-hand side (LHS) of \eqref{eq:toprove} is
\begin{align}
-(n_1+1)\log\left(\frac{\lambda}{\mu}\right) - (e_1+m_1)\log(\gamma) + n_1\log\left(\frac{\lambda}{\mu}\right) + e_1\log(\gamma),
\end{align}
which reduces to 
\begin{align}
-\log\left(\frac{\lambda}{\mu}\right) - m_1\log(\gamma). 
\end{align}

The right-hand side (RHS) of \eqref{eq:toprove} is
\begin{align}
-(n_2+1)\log\left(\frac{\lambda}{\mu}\right) - (e_2+m_2)\log(\gamma) + n_2\log\left(\frac{\lambda}{\mu}\right) + e_2\log(\gamma),
\end{align}
which reduces to
\begin{align}
-\log\left(\frac{\lambda}{\mu}\right) - m_2\log(\gamma). 
\end{align}

Expression \eqref{eq:toprove} reduces to 
\[
-\log\left(\frac{\lambda}{\mu}\right) - m_1\log(\gamma) \leq -\log\left(\frac{\lambda}{\mu}\right) - m_2\log(\gamma). 
\]
Since $\gamma \ge 1$, we know that $\log(\gamma) \ge 0$ and that $m_1 \geq m_2$ by Lemma \ref{lemma:submodular}. Therefore, the LHS of \eqref{eq:toprove} is less than or equal to the RHS of \eqref{eq:toprove} for any $\alpha_1 \subseteq V, \alpha_2 \subseteq \alpha_1, i \in V \setminus \alpha_1$. By definition, $-\log(g(\alpha_{\mathbf{x}}) )$ is a submodular function.
\end{proof}

Theorem~\ref{thm:submodular} proves that $-\log(g(\alpha_{\mathbf{x}}))$ is submodular if $\lambda >0, \mu> 0$, and $\gamma \geq 1$; this means that $\log(g(\alpha_{\mathbf{x}}) )$ is supermodular under the same condition. Since the logarithm function is a monotonic function, the maximum argument of $\log(g(\alpha_{\mathbf{x}}) )$ is also the maximum argument of $g(\alpha_{\mathbf{x}})$, which is the solution to the Most-Probable Configuration Problem. As Regime II) \textbf{ Endogenous Infection Dominant}: $0 < \frac{\lambda}{\mu} \le 1, \gamma >1$ satisfies the condition that $\gamma \ge 1$, using submodular optimization, we can find the \emph{exact} most-probable configuration of the scaled SIS process in Regime II) for \emph{arbitrary} network topology in polynomial time.

\subsection{Social Networks and the Power Grid}
The most-probable configuration allows us to identify the set of agents that are vulnerable to network epidemics since it retains the state of all the agents. Agents who are infected in the most-probable configuration are more vulnerable to the epidemics than agents who remain healthy. Because the most-probable configuration is derived from a dynamical model of network diffusion processes, the set of vulnerable agents depends on the infection and healing rates, $\lambda, \gamma, \mu$. 

As we showed in \cite{JZhangJournal}, the most-probable configuration changes depending on these parameters. When the healing rate is high, $\mathbf{x}^* = \mathbf{x}^0$, meaning that the epidemics is not severe. When the infection rate is high, $\mathbf{x}^* = \mathbf{x}^N$, the epidemics is severe, and all the agents are vulnerable. When $\mathbf{x}^*$ is a non-degenerate configuration (i.e, $\mathbf{x}^* \neq \mathbf{x}^0, \mathbf{x}^N$), this indicates that sets of agents in the network are more vulnerable than others to the epidemics. We illustrate this by solving for the most-probable configuration using \cite{krause2010sfo} under different $\left(\frac{\lambda}{\mu}, \gamma \right)$ parameters for 2 realistic networks: a social network \cite{weeks2002social} and the Western United States power grid \cite{watts1998collective}, obtained from \cite{snapnets}

The network shown in Fig.~\ref{fig:KPP} is a 193 node, 273 edge social network of drug users in Hartford, CT. The network was determined through interviews. Reference \cite{borgatti2003key} looked for influential agents in the network by considering it as a graph connectivity problem. However, they did not consider a dynamical model of influence. Assuming that we can model drug habits as an epidemics (i.e., there is a social contagion aspect to the behavior), we applied the scaled SIS process to this network and solved for the most-probable configuration under different parameters to find influential network structures.

We show the resultant most-probable configurations in Fig.~\ref{fig:KPPlau1_mu5_g2.4}, Fig.~\ref{fig:KPPlau4_mu15_g3}, Fig.~\ref{fig:KPPlau2_mu5_g1.2}, Fig.~\ref{fig:KPPlau1_mu2_g1.6} as we change $\left(\frac{\lambda}{\mu}, \gamma \right)$. We can see from these results that there is a small community of users who are infected when others are healthy. The size of this community increases or decreases depending on the parameters. If there is a social contagion component to drug usage, then these agents may be more vulnerable to the social contagion component of drug usage and therefore more likely to persist in their habit. In the next section, we will relate the most-probable configuration to the network substructure.

The network shown in Fig.~\ref{fig:Power} is the 4941 node, 6595 edge power grid network of the Western United States used by Watts and Strogatz. They showed through simulation of the SIR (susceptible-infected-removed) epidemics model on the western power grid that small-world networks like the western power grid are more conducive to spreading infection/failures than lattice networks. This is useful for explaining why failures propagate so quickly in a blackout. However, they did not identify \emph{which} components in the power grid are more vulnerable to the epidemics.

Figure~\ref{fig:Powerlau1_mu3_g2} and Fig.~\ref{fig:Powerlau1_mu3_g2.6} show the most-probable configuration for the western US power grid when for the scaled SIS process parameterized $\left(\frac{\lambda}{\mu} = 0.33, \gamma = 2 \right)$ and $\left(\frac{\lambda}{\mu} = 0.33, \gamma = 2.6 \right)$, respectively. We can see that for the same $\frac{\lambda}{\mu}$, as $\gamma$ increases, thereby increasing the infectiousness of epidemics, the number of infected agents increases. This is intuitive since, for large $\gamma$, the epidemics is severe, and the most-probable configuration is driven toward $\mathbf{x}^N$, the configuration where all the agents are infected. Moreover, the most-probable configurations are both non-degenerate configurations. The agents who are infected at equilibrium are more vulnerable to the network epidemics than agents who are healthy. By using submodular optimization, we can identify these more vulnerable agents, by solving for the most-probable configuration out of $2^{4941}$ total possible configurations, exactly and in polynomial time.

An important question is to relate the most-probable configuration to network structure. We will show in the next section that the most-probable configuration is related to subgraph density by rewriting the equilibrium distribution \eqref{eq:equilibriumdistribution} in terms of induced subgraphs instead of network configurations.

\begin{figure}[htpb]
        \centering
        \begin{subfigure}[b]{0.48\textwidth}
                \centering
                \includegraphics[width=\textwidth]{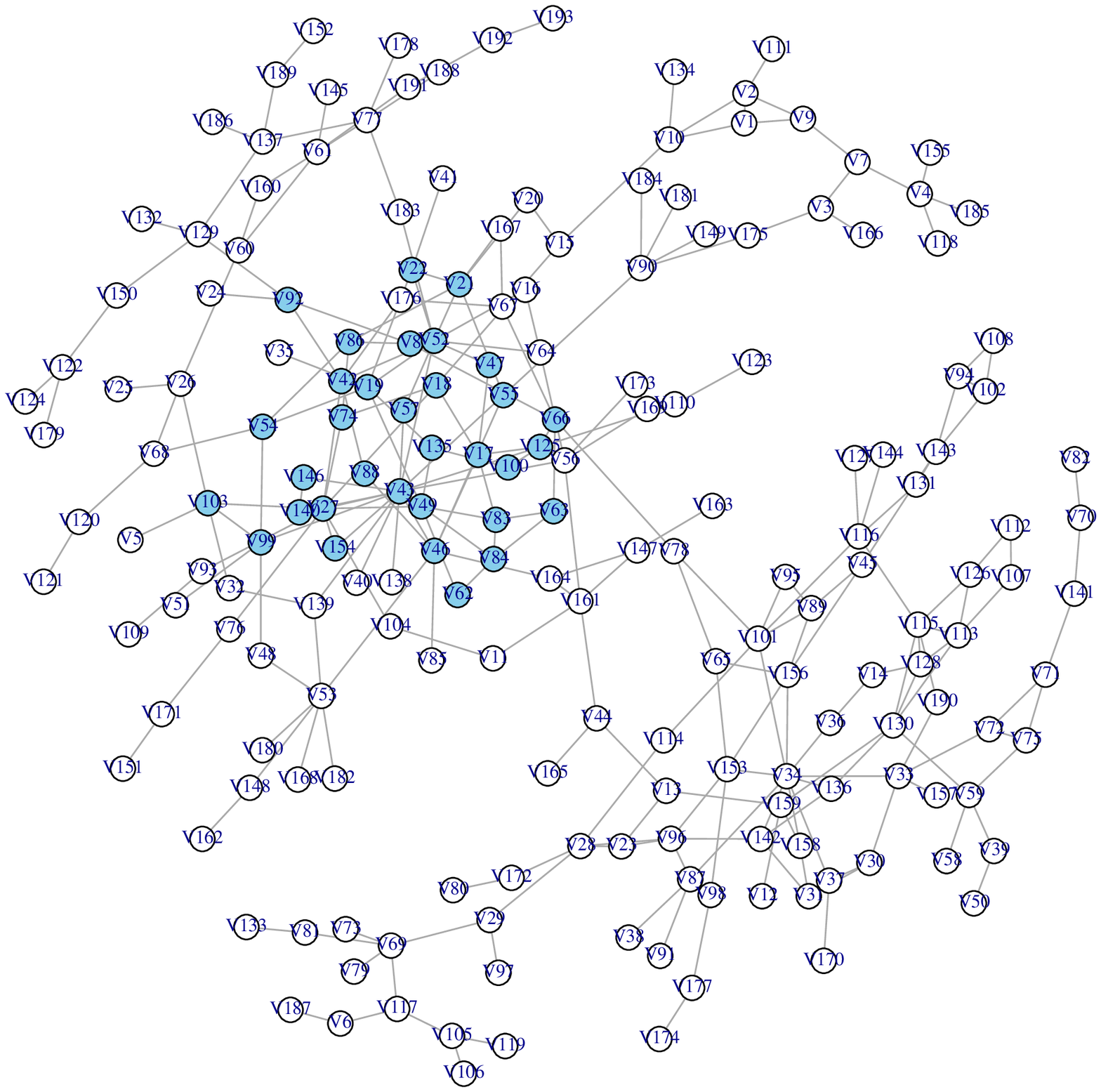}
                \caption{$\frac{\lambda}{\mu} = 0.2, \gamma = 2.4$}%\\ $d(H(\mathbf{x}^*)) = 1.88$}
                \label{fig:KPPlau1_mu5_g2.4}
        \end{subfigure}
        ~ %add desired spacing between images, e. g. ~, \quad, \qquad etc.
        \begin{subfigure}[b]{0.48\textwidth}
                \centering
                \includegraphics[width=\textwidth]{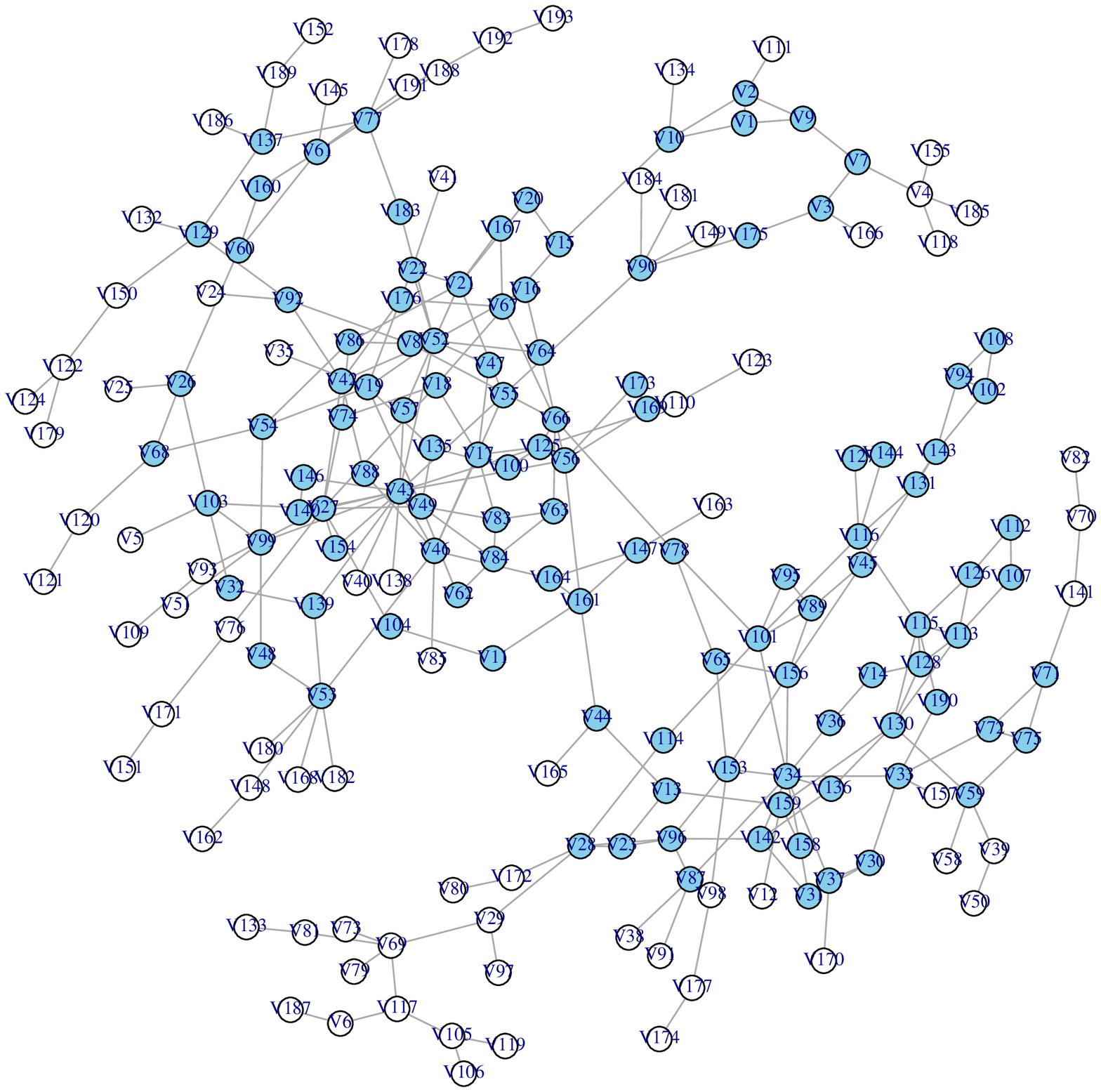}
                \caption{$\frac{\lambda}{\mu} = 0.267, \gamma = 3$}%\\$d(H(\mathbf{x}^*)) = 1.696$}
                \label{fig:KPPlau4_mu15_g3}
        \end{subfigure}
        ~ %add desired spacing between images, e. g. ~, \quad, \qquad etc.
          
        \begin{subfigure}[b]{0.48\textwidth}
                \centering
                \includegraphics[width=\textwidth]{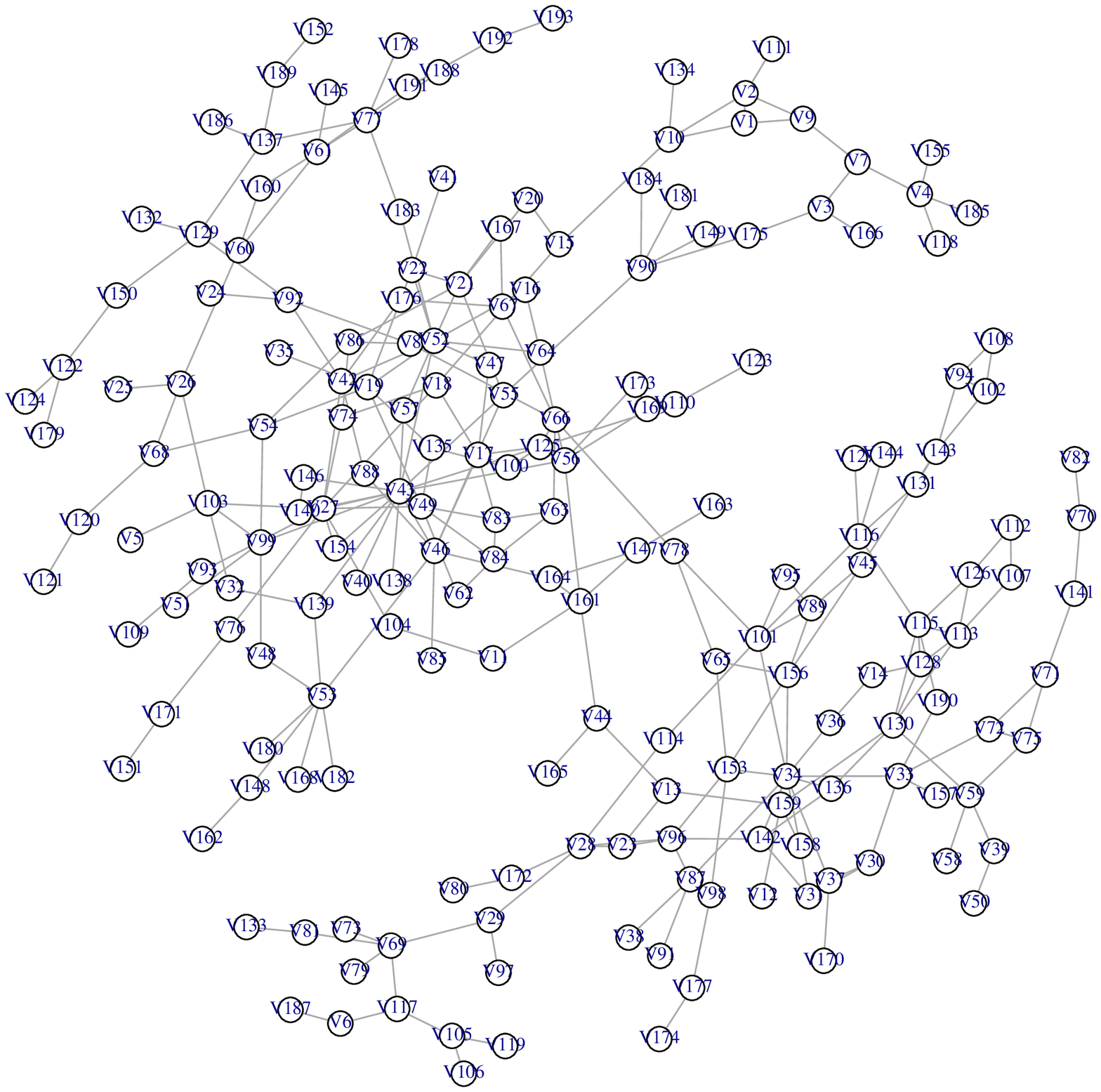}
                \caption{$\frac{\lambda}{\mu} = 0.4, \gamma = 1.2$}%\\$d(H(\mathbf{x}^*)) = 0$}
                \label{fig:KPPlau2_mu5_g1.2}
        \end{subfigure}
            ~ %add desired spacing between images, e. g. ~, \quad, \qquad etc.
        \begin{subfigure}[b]{0.48\textwidth}
                \centering
                \includegraphics[width=\textwidth]{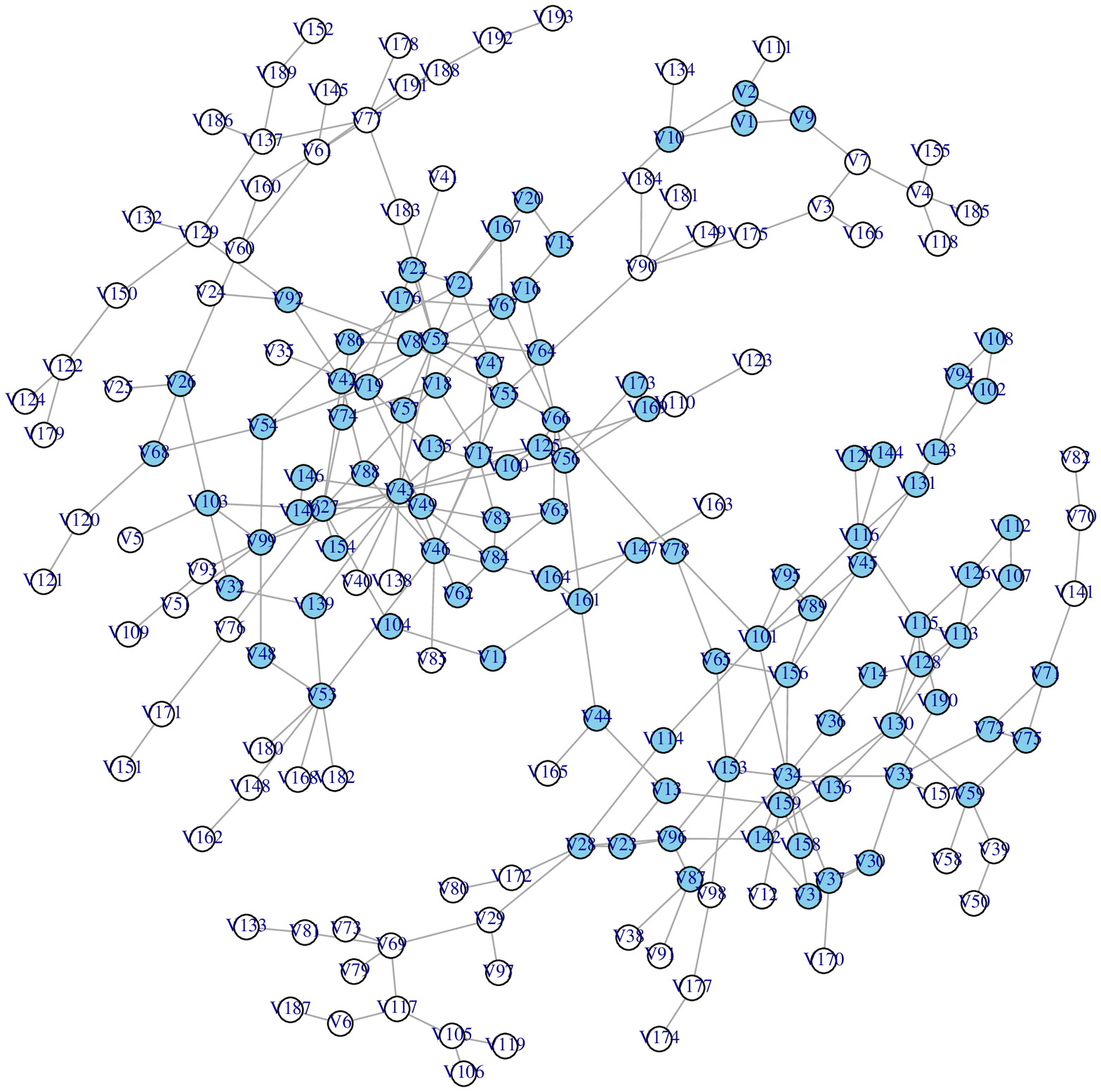}
                \caption{$\frac{\lambda}{\mu} = 0.5, \gamma = 1.6$}%\\$d(H(\mathbf{x}^*)) = 1.73$}
                \label{fig:KPPlau1_mu2_g1.6}
        \end{subfigure}
        \caption{(Color online) Most-Probable Configuration $\mathbf{x}^*$ under Different $\left(\frac{\lambda}{\mu}, \gamma\right)$ Parameters (Blue/Grey = Infected, White = Healthy).}\label{fig:KPP}
\end{figure}

\begin{figure}[htpb]
        \centering
          \begin{subfigure}[b]{0.6\textwidth}
                \centering
                \includegraphics[width=\textwidth]{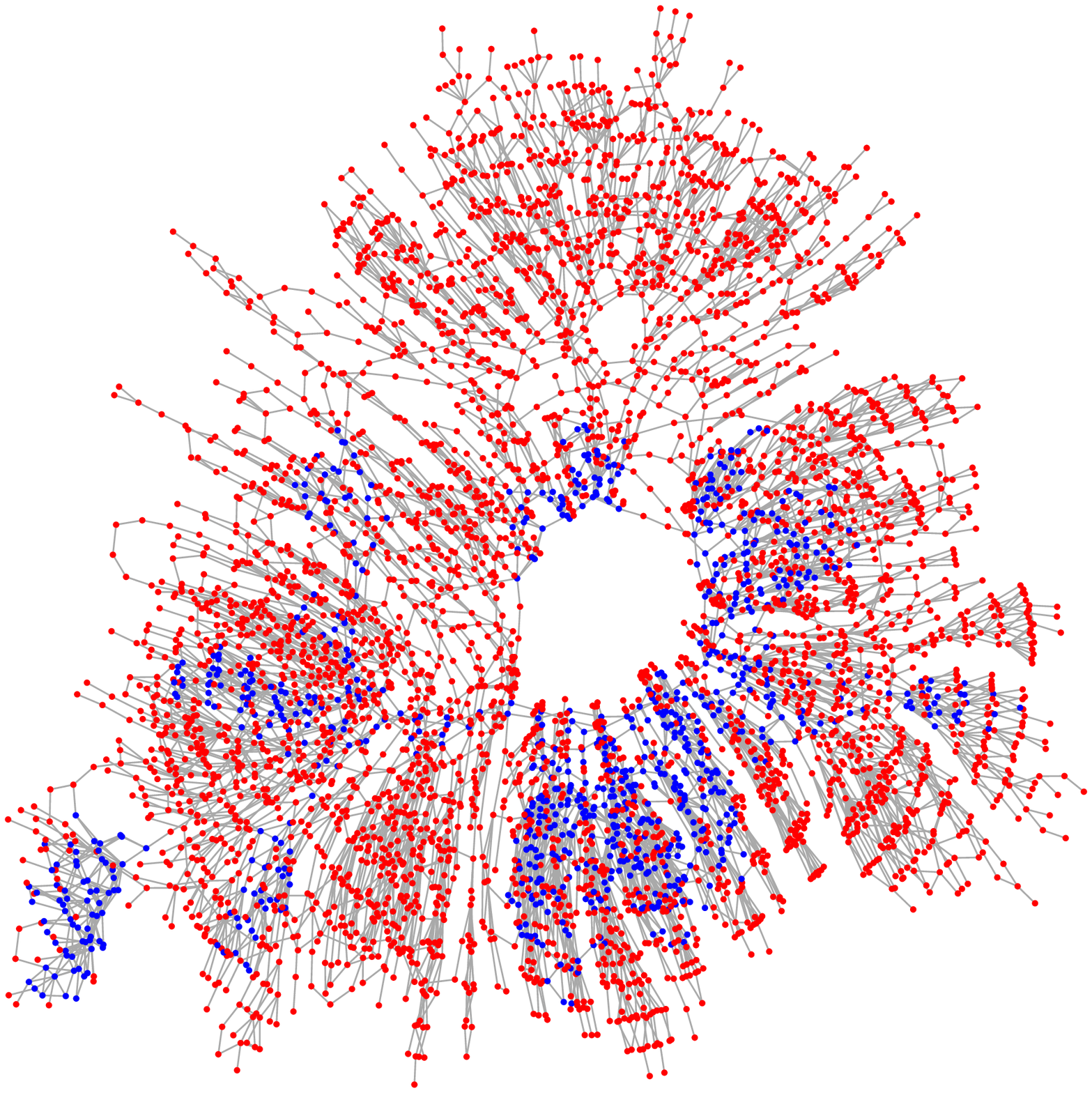}
                \caption{$\frac{\lambda}{\mu} = 0.33, \gamma = 2$}%\\$d(H(\mathbf{x}^*)) = 1.83$}
                \label{fig:Powerlau1_mu3_g2}
        \end{subfigure}
        ~ %add desired spacing between images, e. g. ~, \quad, \qquad etc.
          
        \begin{subfigure}[b]{0.6\textwidth}
                \centering
                \includegraphics[width=\textwidth]{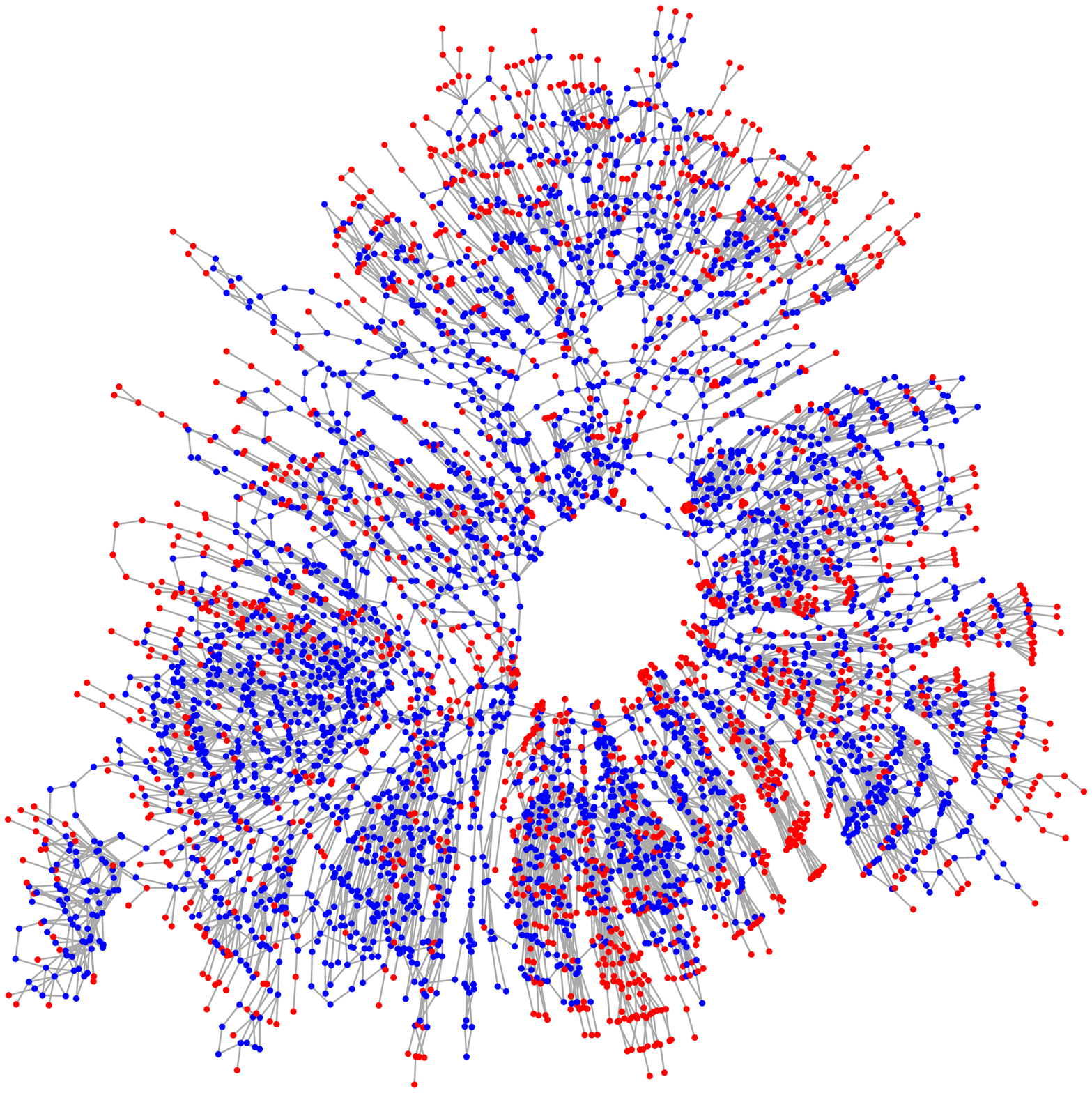}
                \caption{$\frac{\lambda}{\mu} = 0.33, \gamma = 2.6$}%\\$d(H(\mathbf{x}^*)) = 1.5$}
                \label{fig:Powerlau1_mu3_g2.6}
        \end{subfigure}%

        \caption{(Color online) Most-Probable Configuration $\mathbf{x}^*$ under Different $\left(\frac{\lambda}{\mu}, \gamma\right)$ Parameters (Blue/Black = Infected, Red = Healthy).}\label{fig:Power}
\end{figure}

%%%%%%%%%

\section{Most-Probable Configuration and Network Structure}\label{sec:netstruct}

In the previous section, we showed that we can exactly solve for the most-probable configuration with a polynomial time algorithm. The exact solution, however, does not give insight on how the most-probable configuration changes depending on the parameters $\left(\frac{\lambda}{\mu}, \gamma \right)$ and on the network topology. In this section, we draw the connection between the most-probable configuration and subgraphs in the network. As per our intuition for epidemics, densely connected network structures are more vulnerable to network epidemics; the scaled SIS process quantifies this intuition. First, we will define the graph theoretic terms used in this section.

\subsection{Induced Subgraphs and Graph Density}

\begin{definition}[From \cite{algraph}]
The graph $H$ is an induced subgraph of $G$ if two vertices in $H$ are connected if and only if they are connected in $G$  and the vertex set and edge set of $H$ are subsets of the vertex set and edge set of $G$.
\[
V(H) \subseteq V(G), E(H) \subseteq V(G)
\]
\end{definition}

\begin{definition}
The graph $H(\mathbf{x})$ is an induced subgraph of configuration $\mathbf{x} = [x_1, x_2, \ldots, x_N]^T$ if the nodes/edges in the subgraph are the infected agents/edges in $\mathbf{x}$.
\begin{align}\label{eq:subgraphset}
&V(H(\mathbf{x})) = \{v_i \in V(G) \mid x_i = 1 \}\\
&E(H(\mathbf{x})) = \{(i,j) \in E(G) \mid x_i = 1, x_j = 1\}
\end{align}
\end{definition}

By definition, $\left | V(H(\mathbf{x})) \right |= 1^T\mathbf{x}$ and $\left | E(H(\mathbf{x})) \right| = \frac{{\bf x}^TA{\bf x}}{2}$. Figure \ref{fig:configation1andsubgraph} and Fig. \ref{fig:configation2andsubgraph} show two network configurations and their corresponding induced subgraphs. We proved in \cite{JZhang3} that configurations whose induced subgraphs are isomorphic are equally probable. Unless we need to refer explicitly to the underlying network configuration $\mathbf{x}$, for notational simplicity, we will write $H$ to denote an induced subgraph instead of writing $H(\mathbf{x})$.

\begin{figure}[htpb]
\center
\includegraphics[width = 2in]{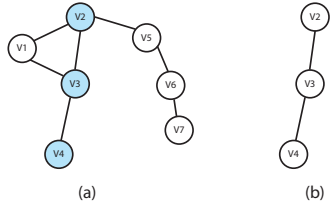}
\caption{(a) Configuration $\mathbf{x}_1 = [0,1,1,1,0,0,0]^T$, (b) Induced Subgraph $H(\mathbf{x}_1)= H_1$}.
\label{fig:configation1andsubgraph}
\end{figure}

\begin{figure}[htpb]
\center
\includegraphics[width = 2in]{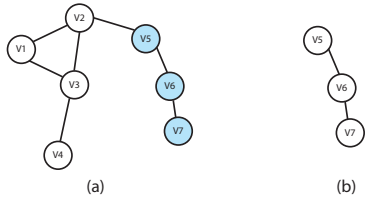}
\caption{(a) Configuration $\mathbf{x}_2 = [0,0,0,0,1,1,1]^T$, (b) Induced Subgraph $H(\mathbf{x}_2) = H_2$}.
\label{fig:configation2andsubgraph}
\end{figure}

\begin{definition}
The set of all possible induced subgraphs of $G$ is $\mathcal{H} = \{H(\mathbf{x})\}, \quad \forall \, \mathbf{x} \in \mathcal{X}$.
\end{definition}

The set $\mathcal{H}$ includes the empty graph, which is induced by the configuration $\mathbf{x}^0= [0,0,\ldots,0]^T$, and $G$, which is the subgraph induced by the configuration $\mathbf{x}^N = [1,1,\ldots,1]^T$. 

\begin{definition}[From \cite{Khuller:2009vj}]
The density of the graph $G$ is 
\[
d(G) = \frac{\left | E(G) \right |}{ \left | V(G) \right |}.
\]
\end{definition}
There is an alternative definition for graph density that is the number of edges divided by the total number of possible edges \cite{wasserman1994social}. Unfortunately, these two definitions of density are not equivalent.

We will refer to the density of the entire network, $d(G) = d(H(\mathbf{x}^N))$, as the \emph{network density}, and the density of an induced subgraph of $G$ as the \emph{subgraph density}. The density of the empty graph, $d(H(\mathbf{x}^0))$, is $0$ by definition. The subgraphs in $\mathcal{H}$ can be partially ordered by their density. There may be many subgraphs with the same density. A special induced subgraph in $\mathcal{H}$ is the densest subgraph. 

\begin{definition}\label{def:densest}
Let $\overline{H}$ be the densest subgraph in $G$. Then
\[
d(\overline{H}) \geq d(H), \quad \forall H \in \mathcal{H}.
\]
\end{definition}
Finding $\overline{H}$ is known as the \emph{Densest Subgraph Problem}. It is known that this problem can be solved in polynomial time exactly and in linear time in approximation for undirected graphs \cite{Khuller:2009vj}.

\subsection{Equilibrium Distribution of the Scaled SIS Process}

Since there is a one-to-one relationship between the network configuration $\mathbf{x}$ and its induced subgraph $H(\mathbf{x})$, we can rewrite the equilibrium distribution \eqref{eq:equilibriumdistribution} of the scaled SIS process in terms of the induced subgraph density and the size of the induced subgraph:
\begin{align}\label{eq:equilibriumdistribution2}
\pi(H) =  \frac{1}{Z}\left(  \left( \frac{\lambda}{\mu}\right)  \gamma^{d(H)}\right)^{\mid V(H) \mid}, \quad H \in \mathcal{H},
\end{align}
where $d(H)$ is the density of the subgraph and $Z$ is the partition function.

The Most-Probable Configuration Problem \eqref{eq:xstar} is then also an optimization problem over all the possible induced subgraphs in $G$:
\begin{align}
H(\mathbf{x}^*) = \arg \max_{H \in \mathcal{H}} \left(  \left( \frac{\lambda}{\mu}\right)  \gamma^{d(H)}\right)^{\mid V(H) \mid}.
\end{align}
The subgraph induced by the most-probable configuration, $H(\mathbf{x}^*)$, is the \emph{most-probable subgraph}, but this is \emph{not} necessarily the same subgraph as the densest subgraph, $\overline{H}$.

Stating the equilibrium distribution in terms of the induced subgraph will allow us to derive several theorems regarding the most-probable configuration. For the theorems that follow, we make the following assumptions:
\begin{description}
\item[Assumption 1.] The scaled SIS process operates in Regime II) \textbf{Endogenous Infection Dominant}. This limits the effective infection and the endogenous infection to the range, $0 < \frac{\lambda}{\mu}\leq 1$ and $\gamma > 1$.
\item[Assumption 2.] The underlying network $G$ is a simple, undirected, unweighted, and connected graph.
\end{description}

%%%%%%%%%%%%%

\subsection{Most-Probable Configuration and Subgraphs}

\begin{theorem}\label{theorem2}[Proof in Appendix \ref{prooftheorem2}]
The most-probable configuration $\mathbf{x}^* \neq \mathbf{x}^0$ if and only if there exists at least one induced subgraph $H \in \mathcal{H}$ with density $d(H)$ for which $\lambda\gamma^{d(H)} > \mu$.
\end{theorem}

\begin{theorem}\label{theorem3}[Proof in Appendix \ref{prooftheorem3}]

Case 1: The densest subgraph, $\overline{H}$, is the network $G$. Then, $\mathbf{x}^* \neq \mathbf{x}^N$ if and only if $\frac{\lambda}{\mu}\gamma^{d(G)} \leq 1$.

Case 2: The densest subgraph, $\overline{H}$, is not the network $G$. Then, $\mathbf{x}^*\neq\mathbf{x}^N$ if and only if there exists at least one induced subgraph $H \in \mathcal{H} \setminus G$ with density $d(H) = \frac{E'}{N'}$ for which
\begin{align}
\frac{\log(\frac{\lambda}{\mu}\gamma^{d(G)})}{\log(\frac{\lambda}{\mu}\gamma^{d(H)} )} < \frac{N'}{N}.
\end{align}

\end{theorem}

\begin{corollary}\label{corollary3}[Proof in Appendix \ref{proofcorollary3}]
Let the density of the network be $d(G) = \frac{E}{N}$. Then, the most-probable configuration is a non-degenerate configuration, $\mathbf{x}^* \in \mathcal{X} \setminus \{\mathbf{x}^0, \mathbf{x}^N\}$, if and only if there exists at least one induced subgraph $H \in \mathcal{H}$ with density $d(H) = \frac{E'}{N'}$ for which $\lambda\gamma^{d(H)} > \mu$, \text{ and}

\[
\frac{\log(\frac{\lambda}{\mu}\gamma^{d(G)})}{\log(\frac{\lambda}{\mu}\gamma^{d(H)} )} < \frac{N'}{N}.
\]

\end{corollary}

In Regime II) individual agents have a preference for being healthy, but the epidemics might spread to other agents through neighbor-to-neighbor contagion. Under the scaled SIS process, the subgraph density $d(H)$ scales the exogenous infection rate $\gamma$, thereby affecting the overall infection rate. Theorem~\ref{theorem2} states that, if the network contains \emph{dense-enough} subgraphs, then even when the effective exogenous infection rate, $\frac{\lambda}{\mu}$, is small (i.e., $0 < \frac{\lambda}{\mu} \ll 1$), the exogenous infection rate, $\gamma$, can leverage dense subgraphs to spread the infection throughout the network.

On the other hand, if the endogenous infection rate, $\gamma$, is large (i.e., $\gamma \gg 1$), then most certainly the epidemics will spread throughout the entire network. Theorem~\ref{theorem3} states when this does not happen. Furthermore, Theorem~\ref{theorem3} shows that it is important to consider if the densest subgraph in the network is the entire network or a smaller subgraph. Corollary~\ref{corollary3} proves that the existence of the non-degenerate configurations is related to the existence of subgraphs with density larger than the network density. The existence of these \emph{denser-than G} subgraphs is crucial to the existence of non-degenerate configurations (i.e., different from $\mathbf{x}^0$ and $\mathbf{x}^N$) as solutions to the Most-Probable Configuration Problem; when the most-probable configuration is a non-degenerate configuration, agents belonging to denser subgraphs are more vulnerable to the epidemics.  

In network science, dense clusters of agents have often been identified as either the network \emph{core} or \emph{community} \cite{csermely2013structure, borgatti2000models, brandes2013studying}. Solving for the non-degenerate configuration is an alternative method for determining these network structures. Previous works in core/community detection are algorithmic and do not consider the dynamical process on the network. The scaled SIS process, however, is a model for dynamical processes on networks and, therefore, what is considered a \emph{community} changes depending on the parameters of the dynamical process: the most-probable configuration changes depending on the exogenous rates $\frac{\lambda}{\mu}$ and on the endogenous rates $\gamma$. 

\begin{figure}[ht]
        \centering
        \begin{subfigure}[b]{0.4\textwidth}
                \centering
                \includegraphics[width=\textwidth]{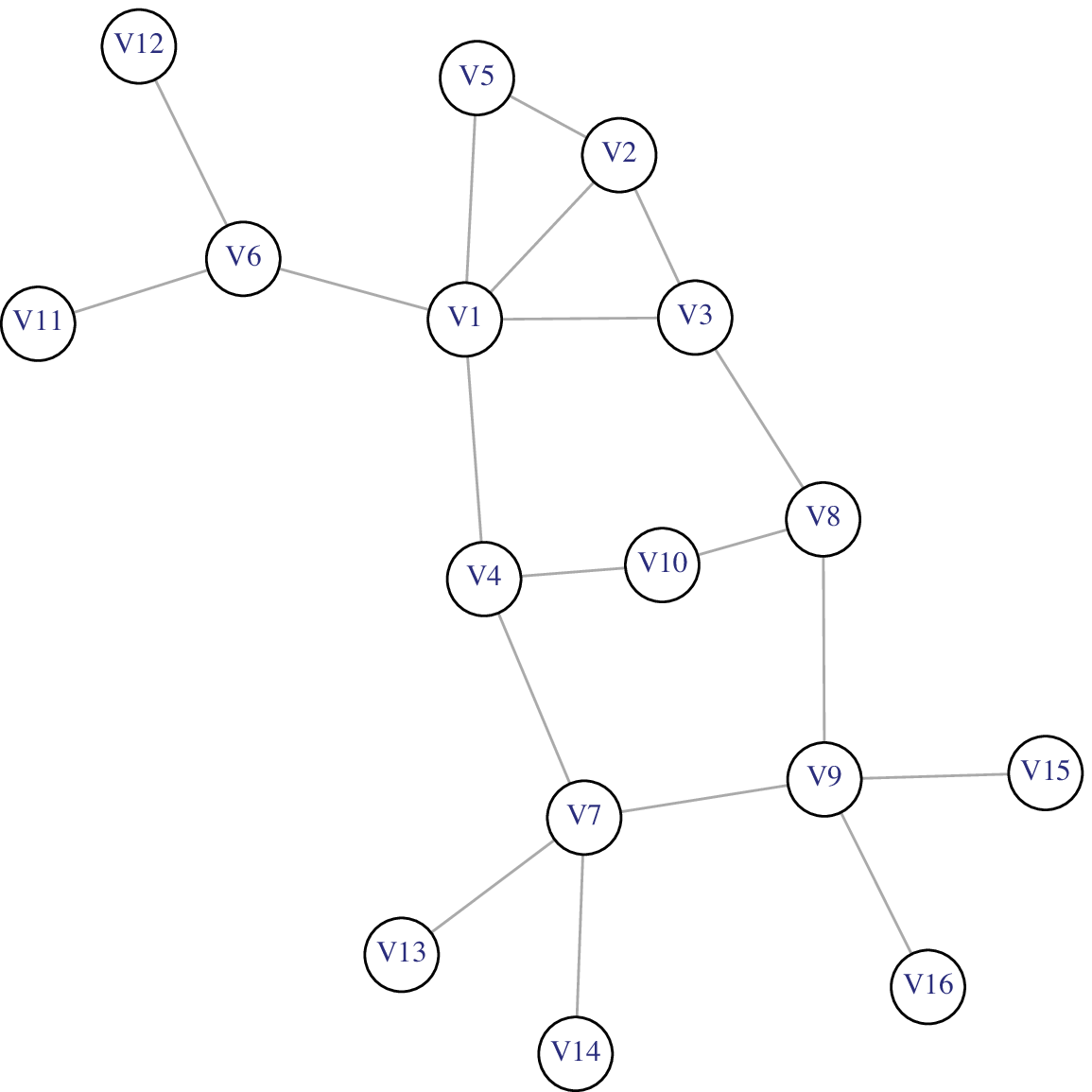}
                \caption{$\frac{\lambda}{\mu} = 0.5, \gamma = 1.2$ \\$d(H(\mathbf{x^*})) = 0$}
                \label{fig:ER16a_lau1_mu2_g1.2}
        \end{subfigure}%
        ~ %add desired spacing between images, e. g. ~, \quad, \qquad etc.
        \begin{subfigure}[b]{0.4\textwidth}
                \centering
                \includegraphics[width=\textwidth]{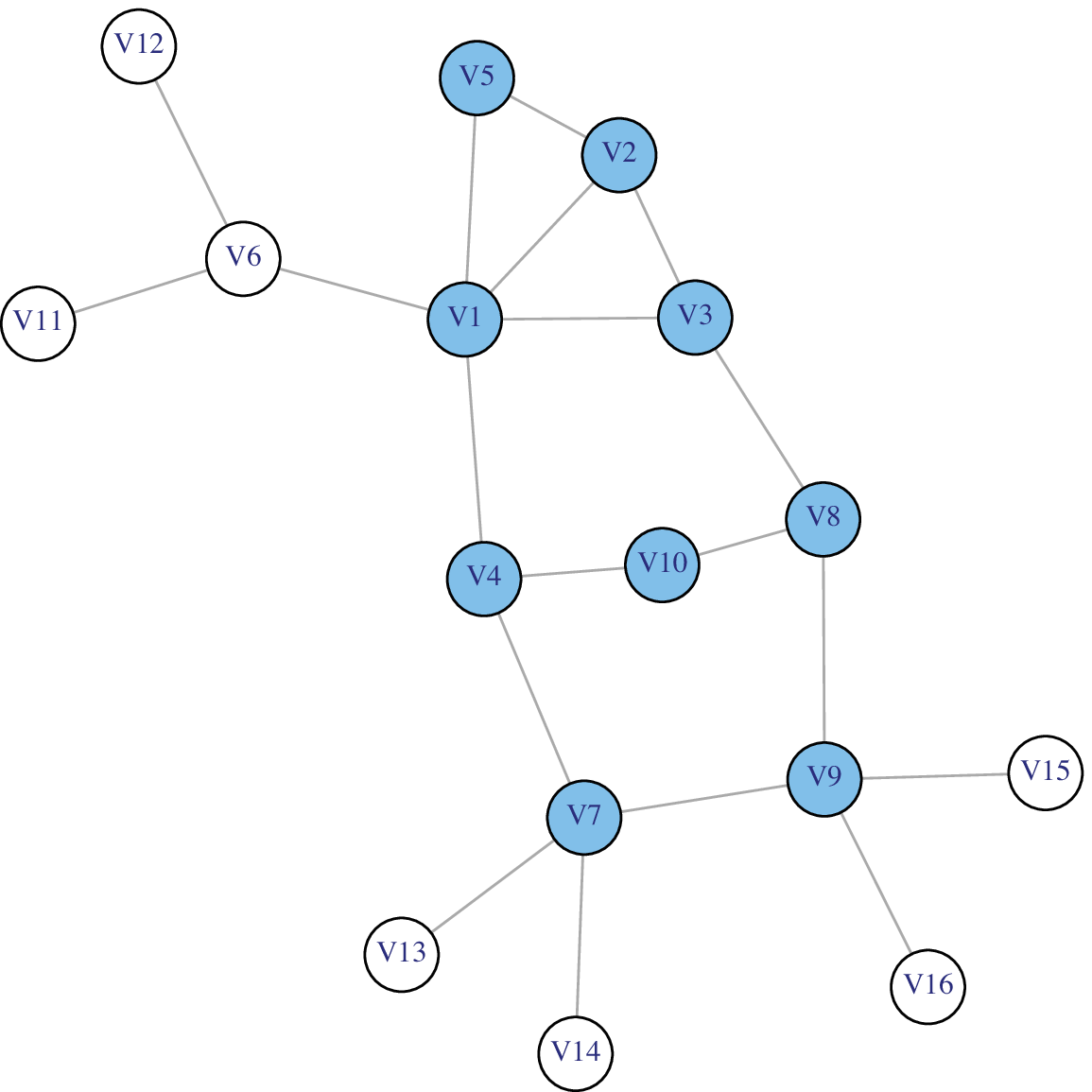}
                \caption{$\frac{\lambda}{\mu} = 0.5, \gamma = 1.7$ \\$d(H(\mathbf{x^*})) = 1.33$}
                \label{fig:ER16a_lau1_mu2_g1.6}
        \end{subfigure}
        
        ~ %add desired spacing between images, e. g. ~, \quad, \qquad etc.
         \begin{subfigure}[b]{0.4\textwidth}
                \centering
                \includegraphics[width=\textwidth]{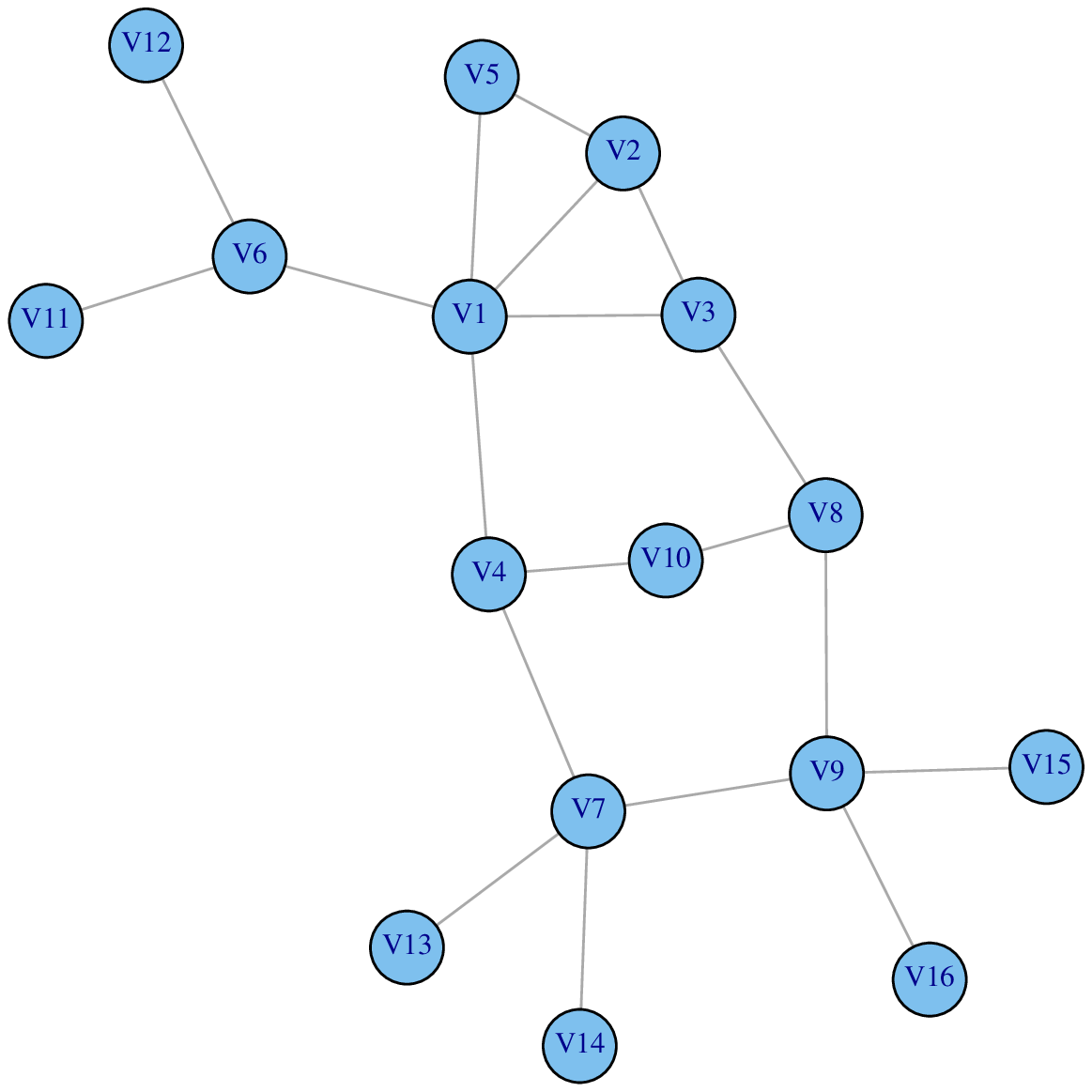}
                \caption{$\frac{\lambda}{\mu} = 0.5, \gamma = 2$ \\$d(H(\mathbf{x^*})) = 1.19$}
                \label{fig:ER16a_lau1_mu2_g2}
        \end{subfigure}
         ~ %add desired spacing between images, e. g. ~, \quad, \qquad etc.
          \begin{subfigure}[b]{0.4\textwidth}
                \centering
                \includegraphics[width=\textwidth]{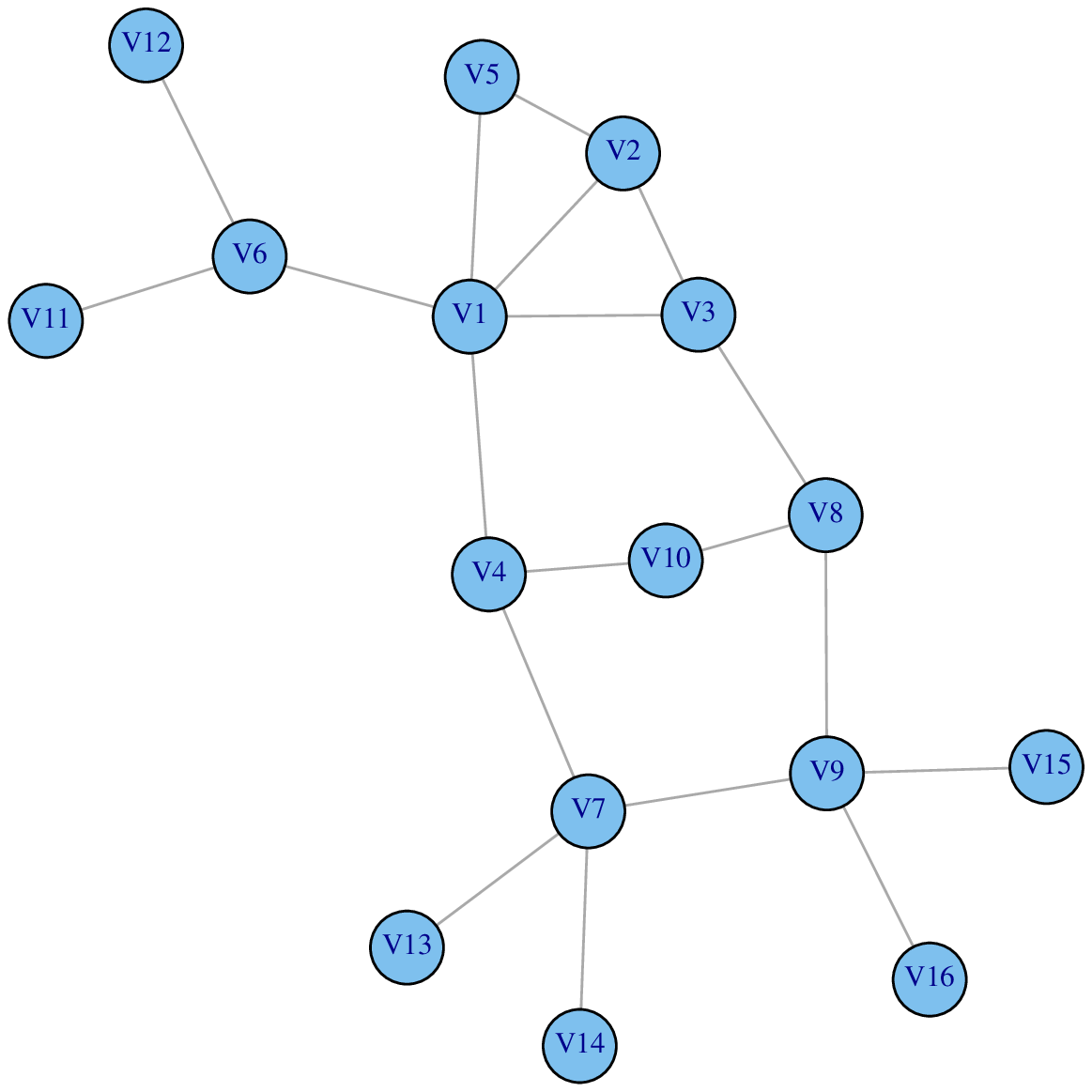}
                \caption{$\frac{\lambda}{\mu} = 0.5, \gamma = 3$ \\$d(H(\mathbf{x^*})) = 1.19$}
                \label{fig:ER16a_lau1_mu2_g3}
        \end{subfigure}
	\caption{(Color online) Most-Probable Configuration $\mathbf{x}^*$ under Different $\left(\frac{\lambda}{\mu}, \gamma\right)$ Parameters (Blue/Grey = Infected, White = Healthy).}\label{fig:ER16a}
\end{figure}

\begin{figure}[ht]
        \centering
        \begin{subfigure}[b]{0.4\textwidth}
                \centering
                \includegraphics[width=\textwidth]{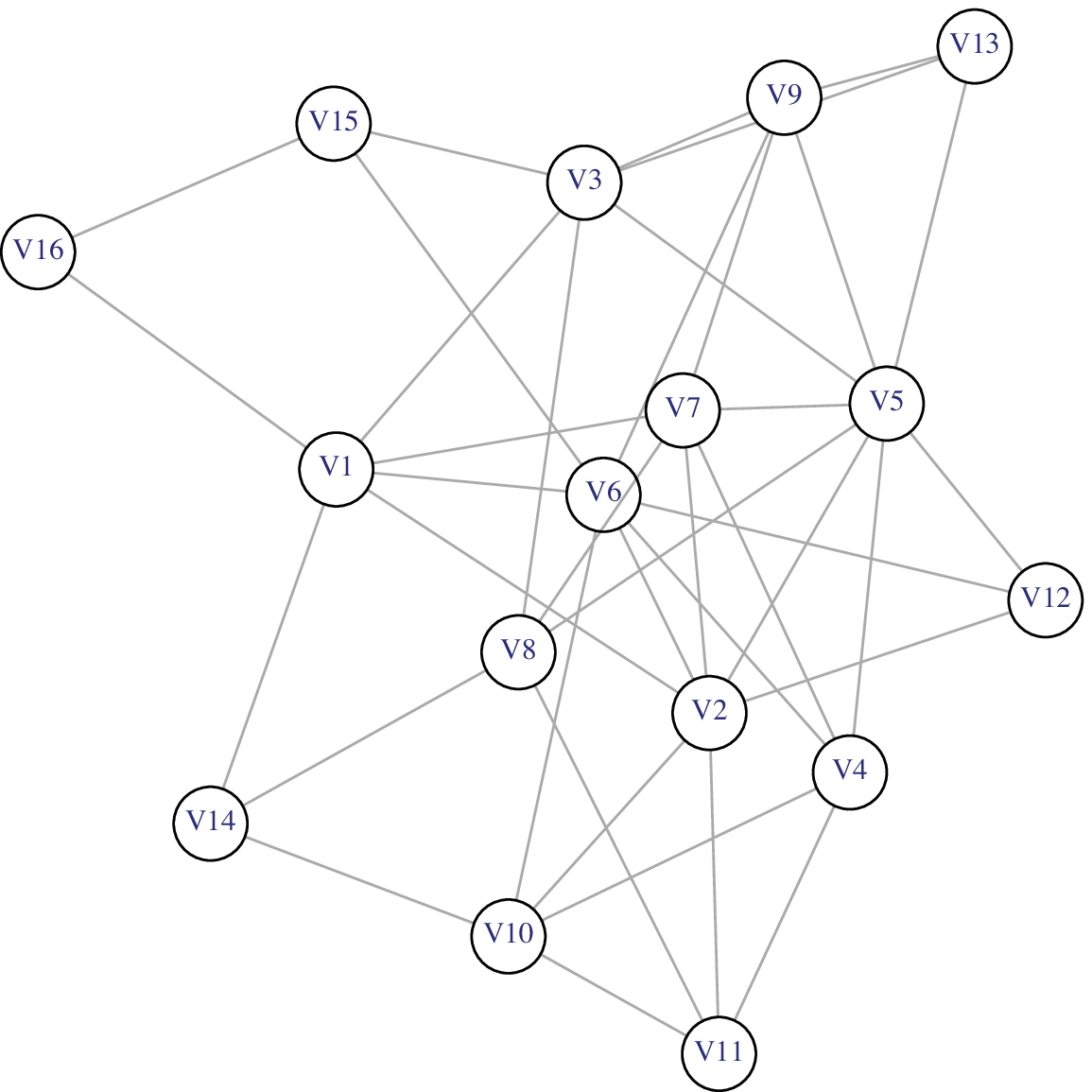}
                \caption{$\frac{\lambda}{\mu} = 0.5, \gamma = 1.2 $\\$d(H(\mathbf{x^*})) = 0$}
                \label{fig:Core16WSc_lau1_mu2_g1.2}
        \end{subfigure}%
        ~ %add desired spacing between images, e. g. ~, \quad, \qquad etc.
        \begin{subfigure}[b]{0.4\textwidth}
                \centering
                \includegraphics[width=\textwidth]{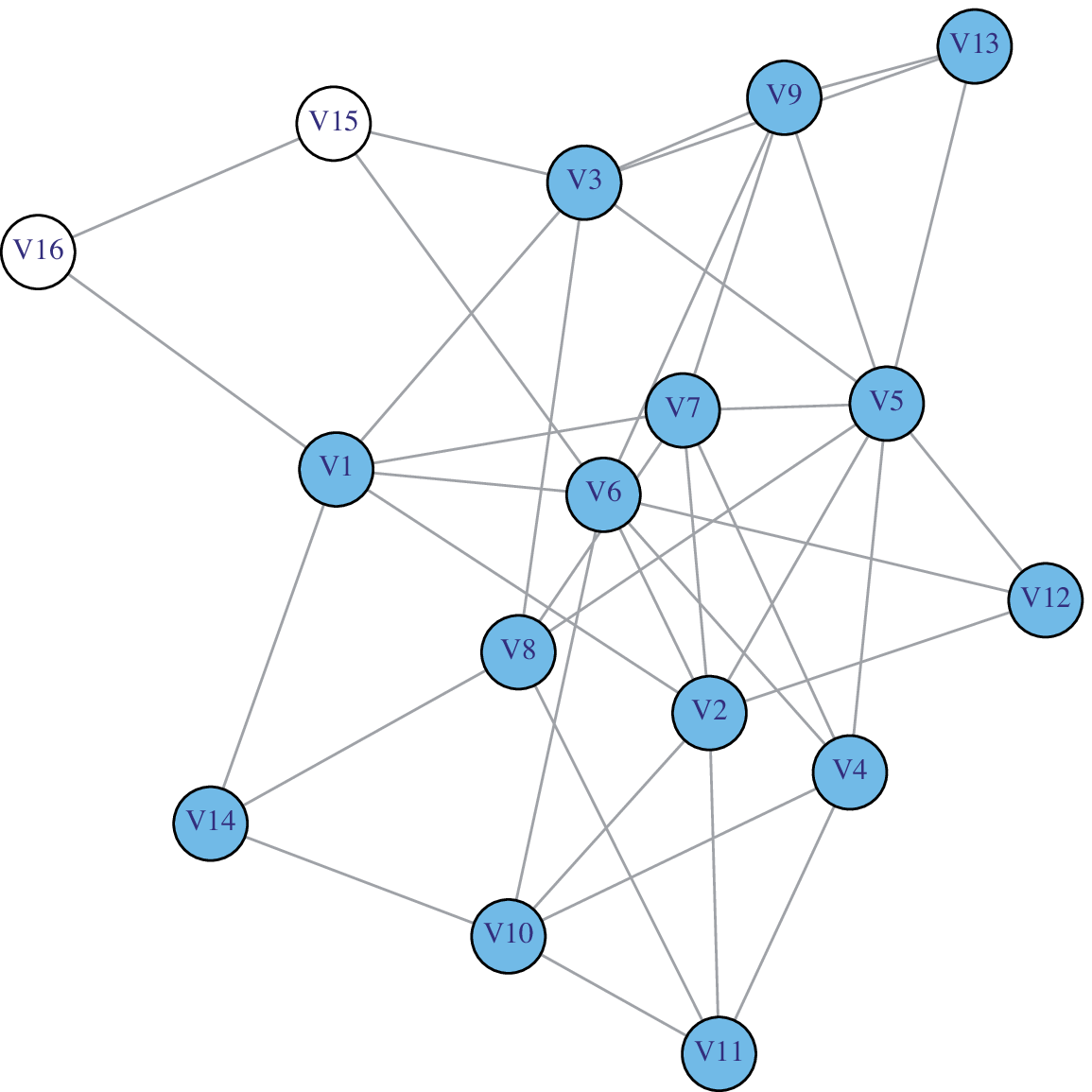}
                \caption{$\frac{\lambda}{\mu} = 0.5, \gamma = 1.38$\\$d(H(\mathbf{x^*})) = 2.5$}
                \label{fig:Core16WSc_lau1_mu2_g1.38}
        \end{subfigure}
        
        ~ %add desired spacing between images, e. g. ~, \quad, \qquad etc.
         \begin{subfigure}[b]{0.4\textwidth}
                \centering
                \includegraphics[width=\textwidth]{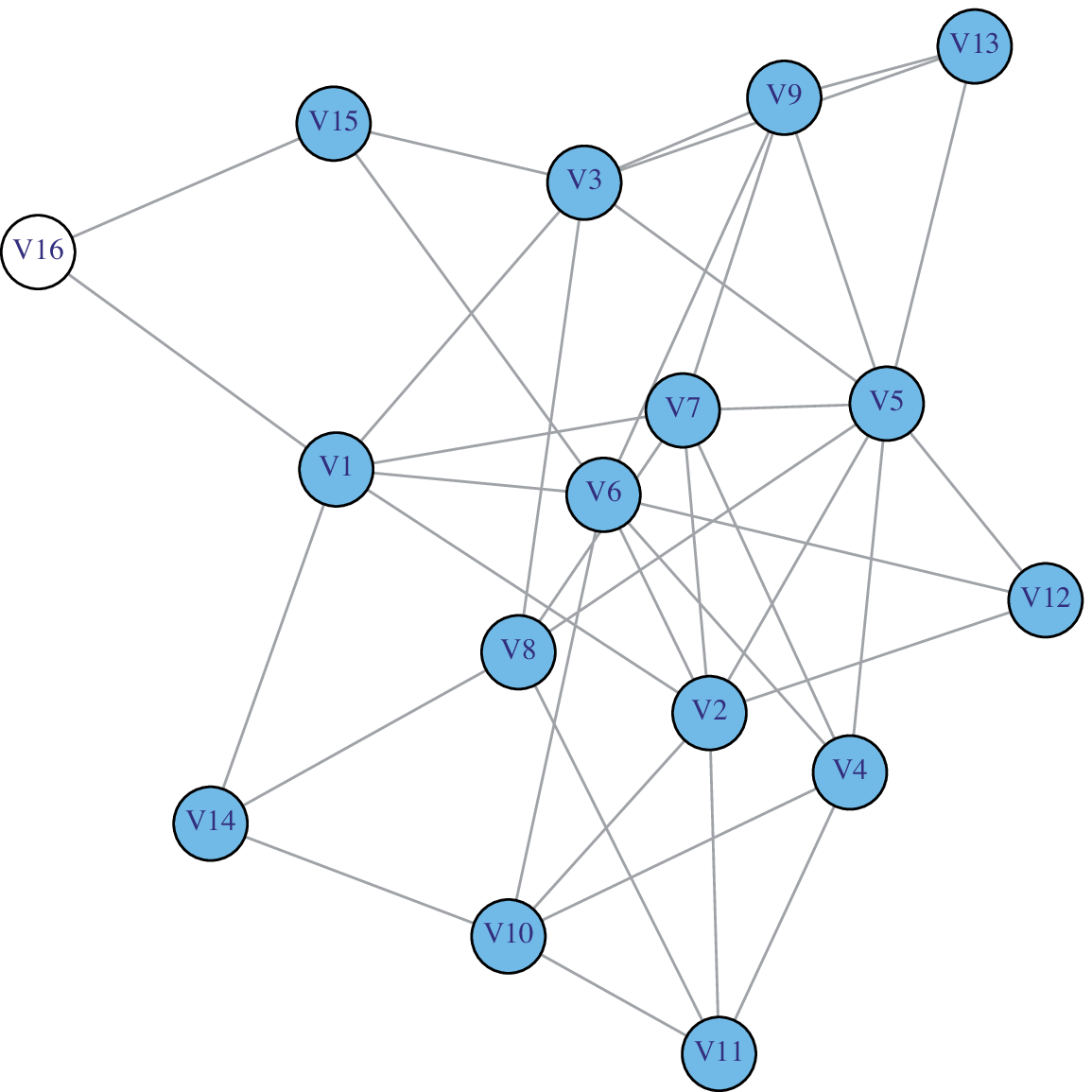}
                \caption{$\frac{\lambda}{\mu} = 0.5, \gamma = 1.41$\\$d(H(\mathbf{x^*})) = 2.467$}
                \label{fig:Core16WSc_lau1_mu2_g1.41}
        \end{subfigure}
         ~ %add desired spacing between images, e. g. ~, \quad, \qquad etc.
          \begin{subfigure}[b]{0.4\textwidth}
                \centering
                \includegraphics[width=\textwidth]{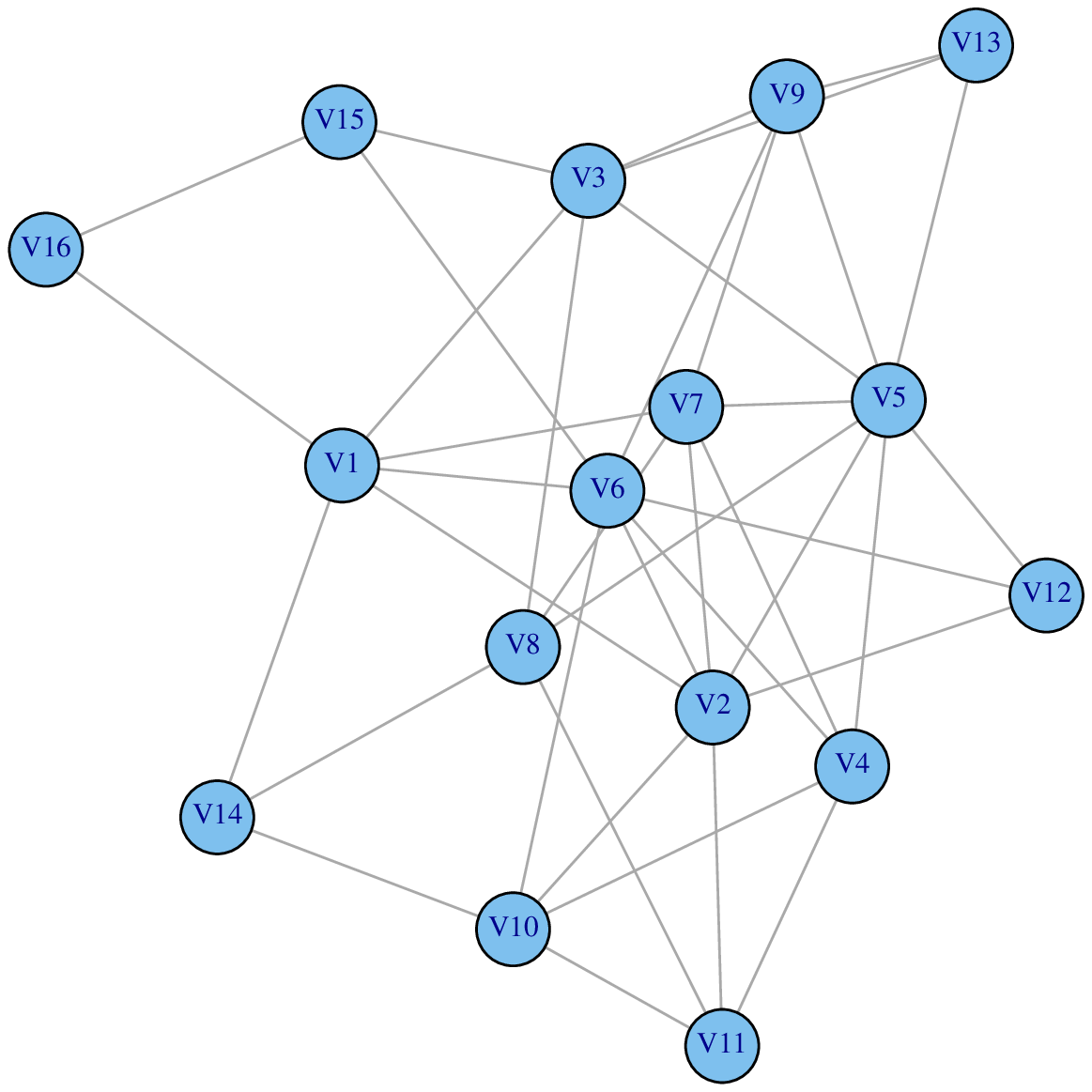}
                \caption{$\frac{\lambda}{\mu} = 0.5, \gamma = 1.7$\\$d(H(\mathbf{x^*})) = 2.4375$}
                \label{fig:Core16WSc_lau1_mu2_g1.6}
        \end{subfigure}
	\caption{(Color online) Most-Probable Configuration $\mathbf{x}^*$ under Different $\left(\frac{\lambda}{\mu}, \gamma\right)$ Parameters (Blue/Grey = Infected, White = Healthy).}\label{fig:Core16WSc}
\end{figure}

To get an easy visual interpretation of Theorem~\ref{theorem2} and~\ref{theorem3}, we illustrate them with two small 16-node examples; Network~A shown in Fig.~\ref{fig:ER16a} and Network~B in  Fig.~\ref{fig:Core16WSc}. For each network, we fix the effective exogenous infection rate, $\frac{\lambda}{\mu} = 0.5$. We then solve for the most-probable configuration for different $\gamma$, ranging from $1.2$ to $3$. As the endogenous infection rate, $\gamma$, changes, the most-probable configuration also changes. In Fig.~\ref{fig:ER16a_lau1_mu2_g1.2} and Fig.~\ref{fig:Core16WSc_lau1_mu2_g1.2}, neither network supports dense enough subgraphs for the epidemics to spread. But as $\gamma$ increases, the infection starts to spread. In Network A, there is at least one subgraph denser than the network. The subgraph induced by $V1, V2, V3, V4, V5, V7, V8, V9, V10$ has a density of 1.33 whereas the density of the entire network is 1.19. In Fig.~\ref{fig:ER16a_lau1_mu2_g1.6}, the most-probable configuration has these 9 agents infected while the other 7 agents remain healthy. The 9 agents in the dense subgraph are more vulnerable to the epidemics when $\frac{\lambda}{\mu} = 0.5$ and $\gamma = 1.7$. 

In Network B, there are at least two subgraphs denser than the network and they are induced by the set of infected agents of the most-probable configuration as shown in Fig.~\ref{fig:Core16WSc_lau1_mu2_g1.38} and Fig.~\ref{fig:Core16WSc_lau1_mu2_g1.41}. We can see by solving for the most-probable configuration for different parameter values that, as the endogenous infection increases, the most-probable configuration goes toward $\mathbf{x}^N$ as all agents become vulnerable to the epidemics.

It is easier for the infection to spread in Network B than in Network A, since, at the same effective exogenous infection rate, $\mathbf{x}^* = \mathbf{x}^N$ for Network B when $\gamma = 1.7$ but $\mathbf{x}^* \neq \mathbf{x}^N$ for Network A with the same exogenous infection rate. This is because Network B is a denser graph ($d(G) = 2.4375) $ than Graph A ($d(G) = 1.19$).

%%%%%%%%%%

\subsection{Most-Probable Configuration and the Densest Subgraph}
We showed that the most-probable configuration is related to the density of induced subgraphs in the network. The densest subgraph, $\overline{H}$, is a special induced subgraph. In this section, we focus specifically on the relationship between the most-probable configuration and the densest subgraph.

\begin{corollary}\label{corollary2}[Proof in Appendix~\ref{proofcorollary2}]
The most-probable configuration $\mathbf{x}^* = \mathbf{x}^0$ if and only if $\lambda \gamma^{d(\overline{H})} \le \mu$.
\end{corollary}

Corollary~\ref{corollary2} follows the result of Theorem \ref{theorem2}. If the densest subgraph in the network is not \emph{dense enough} to overcome individual preferences for being healthy, then the endogenous infection rate $\gamma$ will not be able to drive the most-probable configuration away from $\mathbf{x}^0$.

Lastly, because of the connection between the most-probable configuration of the scaled SIS process and the densest subgraph, we can prove a \emph{general} statement regarding network structure using results from dynamical processes on networks. 

\begin{corollary}\label{corollary4}[Proof in Appendix~\ref{proofcorollary4}]
If $G$ is a $k$-regular, complete multipartite, or complete multipartite with $k$-regular islands network, then $\overline{H} = G$. That is, for these structured networks, the densest subgraph is the overall graph.
\end{corollary}
%Corollary~\ref{corollary4} follows from \ref{corollary3} and our results in \cite{JZhangJournal}.

\section{Conclusion}\label{sec:conclusion}
We introduced in previous works the scaled SIS process, which is a mathematically analyzable model for modeling diffusion processes on a static network \cite{JZhangJournal}. The scaled SIS process is a reversible Markov process and has a closed-form equilibrium distribution that explicitly accounts for the underlying network topology via the adjacency matrix. It is parameterized by 2 parameters: $\left(\frac{\lambda}{\mu}, \gamma \right)$. The effective exogenous infection rate $\frac{\lambda}{\mu}$ controls the exogenous, or the topology-independent behavior of the scaled SIS process whereas the exogenous infection rate $\gamma$ controls the endogenous or the topology-dependent behavior of the process. 

Depending on if the parameter values are between 0 and 1 or great than 1, the scaled SIS process models qualitatively different network diffusion processes. In Regime II) \textbf{Endogenous Infection Dominant:} $0 < \frac{\lambda}{\mu} \leq 1, \gamma > 1$, the scaled SIS process best models a network epidemics process; individuals prefer to be healthy, while neighbor-to-neighbor infection helps to spread the epidemics throughout the population.

This paper analyzes the Most-Probable Configuration Problem that solves for the network state with the maximum equilibrium probability, in Regime II) for arbitrary networks. First, we prove that the Most-Probable Configuration Problem in Regime II) is submodular. This means that we can compute the \emph{exact} most-probable configuration in \emph{polynomial time}. We use the most-probable configuration of the scaled SIS process to identify sets of vulnerable agents/components for a social network of drug users and the Western US power grid under different infection/healing rates.

We then showed that the most-probable configuration is dependent on certain classes of subgraphs in the networks. If there exist \emph{dense-enough} subgraphs, conditioned on the right set of parameters, the most-probable configuration will shift away from $\mathbf{x}^0$, the network state where all the agents are healthy. However, if there exist subgraphs that are \emph{denser-than} the entire network, conditioned on the right range of infection and healing rates, the most-probable configuration may not reach $\mathbf{x}^N$, the network state with all agents infected. We call the solution of the Most-Probable Configuration Problem that is neither $\mathbf{x}^0$ nor $\mathbf{x}^N$, the non-degenerate configuration. Non-degenerate configurations identify subsets of agents that are more vulnerable to the network epidemics than others. 

We also proved in this paper using results in \cite{JZhangJournal} that structured networks such as $k$-regular, complete multipartite, complete multipartite with $k$-regular islands do not contain subgraphs that are denser than the overall network. Therefore, if we want to avoid subsets of agents being more vulnerable than others, we should use these types of structured networks. 

Our analysis of the scaled SIS process in Regime II) informs us that network subgraph structures are important for understanding network diffusion processes. For future work, we are interested in statistically characterizing the subgraphs in network classes such as small-world networks and scaled-free networks.

\begin{acknowledgments}
This work is partially supported by AFOSR grant FA95501010291, and by NSF grants CCF1011903 and CCF1018509. We wish to thank Prof. Jo{\~a}o P. Costeira and Prof. Jo{\~a}o M.F. Xavier of the Department of Electrical and Computer Engineering at Instituto Superior T{\'e}cnico, Lisbon, Portugal, for discussions regarding submodular optimization.
\end{acknowledgments}

%%%%%%%%%%%
\appendix

\section{Proof for Theorem~\ref{theorem2}}\label{prooftheorem2}
\begin{theorem*}
The most-probable configuration $\mathbf{x}^* \neq \mathbf{x}^0$ if and only if there exists at least one induced subgraph $H \in \mathcal{H}$ with density $d(H)$ for which $\lambda\gamma^{d(H)} > \mu$.
\end{theorem*}

\begin{proof}

\textbf{Sufficiency:} If there exists at least one subgraph $H \in \mathcal{H}$ with density $d(H)$ for which $\lambda\gamma^{d(H)} > \mu$, then $\mathbf{x}^* \neq \mathbf{x}^0$.

Using the equilibrium distribution \eqref{eq:equilibriumdistribution}, $\pi(\mathbf{x}^0) = \frac{1}{Z}$. Let the subgraph $H \in \mathcal{H}$ be the subgraph induced by configuration $\mathbf{x}' \in \mathcal{X} \setminus \mathbf{x}^0$. The number of infected agents in configuration $\mathbf{x}'$ is $1^T\mathbf{x}' = \mid V(H) \mid > 0$. Using \eqref{eq:equilibriumdistribution2}, its equilibrium probability is 
\[
\pi(\mathbf{x}') = \pi(H) = \frac{1}{Z}\left(  \left( \frac{\lambda}{\mu}\right)  \gamma^{d(H)} \right)^{\mid V(H) \mid}
\]

If $ \left( \frac{\lambda}{\mu}\right)\gamma^{d(H)} > 1$, we know that $\pi(\mathbf{x}') > \pi(\mathbf{x}^0)$. Therefore, $\mathbf{x}^0$ can not be the most-probable configuration.

\textbf{Necessity:} If $\mathbf{x}^* \neq \mathbf{x}^0$, then there exist at least one subgraph $H \in \mathcal{H}$ with density $d(H)$ for which $\lambda\gamma^{d(H)} > \mu$.

If $\mathbf{x}^* \neq \mathbf{x}^0$, this means that there is some configuration $\mathbf{x}'$ for which $\pi({\bf x'}) > \pi({\bf x}^0)$. We know that $\pi({\bf x}^0) = \frac{1}{Z}$. Using the equilibrium distribution in \eqref{eq:equilibriumdistribution2} and the fact that ${1^T{\bf x}} = |V(H)| > 0, \, \forall \, \mathbf{x} \in \mathcal{X} \setminus  {\bf x}^0 $, we can conclude that there must exist some induced subgraph whose density satisfies this condition $\left( \frac{\lambda}{\mu}\right)  \gamma^{d(H(\mathbf{x'})}  >1$.
\end{proof}

%%%%%%%%%%%theorem 3 proof

\section{Proof for Theorem~\ref{theorem3}}\label{prooftheorem3}
\begin{theorem*}
Case 1: The densest subgraph, $\overline{H}$, is the network $G$. Then, $\mathbf{x}^* \neq \mathbf{x}^N$ if and only if $\frac{\lambda}{\mu}\gamma^{d(G)} \leq 1$.

Case 2: The densest subgraph, $\overline{H}$, is not the network $G$. Then, $\mathbf{x}^*\neq\mathbf{x}^N$ if and only if there exists at least one induced subgraph $H \in \mathcal{H} \setminus G$ with density $d(H) = \frac{E'}{N'}$ for which
\begin{align}\label{eq:condition3}
\frac{\log(\frac{\lambda}{\mu}\gamma^{d(G)})}{\log(\frac{\lambda}{\mu}\gamma^{d(H)} )} < \frac{N'}{N}.
\end{align}
\end{theorem*}

\begin{proof}

\textbf{Sufficiency:} Lets first prove sufficiency for both case 1 and case 2.

Case 1: $\overline{H} = G$. If $\lambda\gamma^{d(G)} \leq \mu$, then $\mathbf{x}^* \neq \mathbf{x}^N$.  

Follows from Corollary~\ref{corollary2}: If $\lambda \gamma^{d(\overline{H}(\mathbf{x}))} \le \mu$, then $\mathbf{x}^* = \mathbf{x}^0$.

Case 2: $\overline{H} \neq G$. If there exists at least one induced subgraph $H \in \mathcal{H}$ with density $d(H) = \frac{E'}{N'}$ such that $\frac{\log(\frac{\lambda}{\mu}\gamma^{d(G)})}{\log(\frac{\lambda}{\mu}\gamma^{d(H)} )} < \frac{N'}{N}$, then $\mathbf{x}^*\neq\mathbf{x}^N$.

The subgraph $H$ is induced by the configuration $\mathbf{x}' \in \mathcal{X}$. The log equilibrium probability according to \eqref{eq:equilibriumdistribution2} for $\mathbf{x}'$ and $\mathbf{x}^N$, respectively, are:
\[
\log(\pi(\mathbf{x}')) = \log\left(\frac{1}{Z}\right) + N' \log\left(\frac{\lambda}{\mu}\gamma^{d(H)} \right)
\] 
and
\[
\log(\pi(\mathbf{x}^N)) = \log\left(\frac{1}{Z}\right) + N \log\left(\frac{\lambda}{\mu}\gamma^{d(G)} \right).
\]

Condition $\frac{\log(\frac{\lambda}{\mu}\gamma^{d(G)})}{\log(\frac{\lambda}{\mu}\gamma^{d(H)} )} < \frac{N'}{N}$ implies that $N \log\left(\frac{\lambda}{\mu}\gamma^{d(G)}\right)  < N' \log\left(\frac{\lambda}{\mu}\gamma^{d(H)}\right)$. Therefore, $\log(\pi(\mathbf{x}')) >  \log(\pi(\mathbf{x}^N))$. Since the logarithm is a monotonic function, we can conclude that $\mathbf{x}^*\neq\mathbf{x}^N$.

\textbf{Necessity:} We now prove necessity for both case 1 and case 2.

Case 1: $\overline{H} = G$. If $\mathbf{x}^* \neq \mathbf{x}^N$, then $\lambda\gamma^{d(G)} \leq \mu$.

Follows from Corollary~\ref{corollary2}: If $\mathbf{x}^* = \mathbf{x}^0$, then $\lambda \gamma^{d(\overline{H}(\mathbf{x}))} \le \mu$.

Case 2: $\overline{H} \neq G$. If $\mathbf{x}^* \neq \mathbf{x}^N$, then there exists at least one induced subgraph $H \in \mathcal{H}$ such that $\frac{\log(\frac{\lambda}{\mu}\gamma^{d(G)})}{\log(\frac{\lambda}{\mu}\gamma^{d(H)} )} < \frac{N'}{N}$.

Let $\mathbf{x}^* = \mathbf{x}'$, which induces a subgraph $H \in \mathcal{H}$ with density $d(H)$. Using \eqref{eq:equilibriumdistribution2},
\[
\pi(\mathbf{x}') = \log\left(\frac{1}{Z}\right) + N'\log\left( \frac{\lambda}{\mu}  \gamma^{d(H)} \right)
\]
\[
\pi(\mathbf{x}^N) = \log\left(\frac{1}{Z}\right) + N\log\left( \frac{\lambda}{\mu}  \gamma^{d(G)} \right).
\]
This means $\pi(\mathbf{x}') - \pi(\mathbf{x}^N) > 0$, which implies
\[
N'\log\left( \frac{\lambda}{\mu}  \gamma^{d(H)} \right)  - N\log\left( \frac{\lambda}{\mu}  \gamma^{d(G)} \right) > 0
\]
This reduces to the condition that 
\[
\frac{\log(\frac{\lambda}{\mu}\gamma^{d(G)})}{\log(\frac{\lambda}{\mu}\gamma^{d(H)} )} < \frac{N'}{N}.
\]

\end{proof}

%%%%%%%%%%%%%%%%%%%%%%%%%%%%%COROLLARYSPURIOUS
\section{Proof for Corollary~\ref{corollary3}}\label{proofcorollary3}

\begin{corollary*}
Let the density of the network be $d(G) = \frac{E}{N}$. Then, the most-probable configuration is a non-degenerate configuration, $\mathbf{x}^* \in \mathcal{X} \setminus \{\mathbf{x}^0, \mathbf{x}^N\}$, if and only if there exists at least one induced subgraph $H \in \mathcal{H}$ with density $d(H) = \frac{E'}{N'}$ for which $\lambda\gamma^{d(H)} > \mu$, \text{ and}

\[
\frac{\log(\frac{\lambda}{\mu}\gamma^{d(G)})}{\log(\frac{\lambda}{\mu}\gamma^{d(H)} )} < \frac{N'}{N}.
\]
\end{corollary*}

\begin{proof}

Theorem~\ref{theorem2} gives the necessary and sufficient condition for the most-probable configuration $\mathbf{x}^* \neq \mathbf{x}^0$ to be existence of a subgraph $H$ such that $\lambda\gamma^{d(H)} > \mu$. Theorem~\ref{theorem3} gives the necessary and sufficient condition that the most-probable configuration is not $\mathbf{x}^N$ when 
\[
\frac{\log(\frac{\lambda}{\mu}\gamma^{d(G)})}{\log(\frac{\lambda}{\mu}\gamma^{d(H)} )} < \frac{N'}{N}.
\]
This proves the Corollary.
\end{proof}

%%%%%%%%%%%%%%%%
\section{Proof for Corollary~\ref{corollary2}}\label{proofcorollary2}

\begin{corollary*}
The most-probable configuration $\mathbf{x}^* = \mathbf{x}^0$ if and only if $\lambda \gamma^{d(\overline{H})} \le \mu$.
\end{corollary*}

\begin{proof}

\textbf{Sufficiency:} If $\lambda \gamma^{d(\overline{H})} \leq \mu$, then $\mathbf{x}^* = \mathbf{x}^0$.

Recall the definition of the densest subgraph \ref{def:densest}. With $\gamma > 1$, $\lambda \gamma^{d(H(\mathbf{x}))} \le\lambda \gamma^{d(\overline{H}(\mathbf{x}))} \le \mu $ for all possible induced subgraphs in $G$. This means that there is no subgraph, $H \in \mathcal{H}$, for which $\lambda \gamma^{d(H)} > \mu$. We can conclude that $\mathbf{x}^* = \mathbf{x}^0$ using the contrapositive of Theorem~\ref{theorem2}: If there is no subgraph $H \in \mathcal{H}$ with density $d(H)$ for which $\lambda\gamma^{d(H)} > \mu$, then $\mathbf{x}^* = \mathbf{x}^0$.

\textbf{Necessity:} If $\mathbf{x}^* = \mathbf{x}^0$, then $\lambda \gamma^{d(\overline{H})} \le \mu$.

The result follows from the contrapositive of Theorem~\ref{theorem2}: If $\mathbf{x}^* = \mathbf{x}^0$, then there is no subgraph $H \in \mathcal{H}$ with density $d(H)$ for which $\lambda\gamma^{d(H)} > \mu$. Therefore, all induced subgraphs, including the densest subgraph have density for which $\lambda\gamma^{d(H)} \le \mu$.

\end{proof}

%%%%%%%%%%%%%%%%%%%%%DENSEST SUBGRAPH PROOF%%%%

\section{Proof for Corollary~\ref{corollary4}}\label{proofcorollary4}

\begin{corollary*}
If $G$ is a $k$-regular, complete multipartite, or complete multipartite with $k$-regular islands network, then $\overline{H} = G$. That is, for these structured networks, the densest subgraph is the overall graph.
\end{corollary*}

\begin{proof}
We proved previously in \cite{JZhangJournal} that the solution of the Most-Probable Configuration Problem for any parameters $\left(\frac{\lambda}{\mu}, \gamma \right)$ in Regime II) \textbf{Endogenous Infection Dominant:} $0 < \frac{\lambda}{\mu} \le 1, \gamma > 1$, over $k$-regular, complete multipartite, complete multipartite with $k$-regular islands networks is either $\bf{x}^0$ and/or $\bf{x}^N$; the solution to the Most-Probable Configuration Problem for these networks is not a non-degenerate configuration in Regime II). We will use this and Corollary~\ref{corollary3} to prove this corollary.

Consider the contrapositive of Corollary~\ref{corollary3}: Let the density of the network be $d(G) = \frac{E}{N}$. Then, the most-probable configuration is not a non-degenerate configuration, $\mathbf{x}^* \in \{\mathbf{x}^0, \mathbf{x}^N\}$, if and only if there does not exist any subgraph $H \in \mathcal{H}$ with density $d(H) = \frac{E'}{N'}$ for which $\lambda\gamma^{d(H)} > \mu$, \text{ or}
\[
\frac{\log(\frac{\lambda}{\mu}\gamma^{d(G)})}{\log(\frac{\lambda}{\mu}\gamma^{d(H)} )} < \frac{N'}{N}.
\]

This implies that all the induced subgraphs, $H \in \mathcal{H}$, in networks whose solution to the Most-Probable Configuration Problem is not a non-degenerate configuration in Regime II), satisfy the condition that $\lambda\gamma^{d(H)} \le \mu$ or 
\[
\frac{\log(\frac{\lambda}{\mu}\gamma^{d(G)})}{\log(\frac{\lambda}{\mu}\gamma^{d(H)} )} \geq \frac{N'}{N},
\]
for all $0 < \frac{\lambda}{\mu} \le 1, \gamma > 1$.

Depending on the effective infection rate and the endogenous infection rate, $\left(\frac{\lambda}{\mu}, \gamma \right)$, the first condition $\lambda\gamma^{d(H)} \le \mu$ may not be satisfied. However, since $\frac{N'}{N} \leq 1$ regardless of the parameters and the underlying network, the second condition is satisfied if
\[
\frac{\log(\frac{\lambda}{\mu}\gamma^{d(G)})}{\log(\frac{\lambda}{\mu}\gamma^{d(H)} )} \geq 1, \quad \forall \, H \in \mathcal{H}.
\]
Since $\gamma > 1$, this means that $d(H) \leq d(G)$ for all possible induced subgraph. As this only depend on the structure of the underlying network, we can conclude that $d(H) \leq d(G)$ for networks whose most-probable configuration can only be $\mathbf{x}^0$ and/or $\mathbf{x}^N$.

\end{proof}

%%%%BIBLIOGRPAPHY
\bibliographystyle{apsrev4-1}
\bibliography{./Main_text_incl._figuresNotes}

\end{document}